\documentclass[]{article}

\usepackage{amsmath, amssymb, dsfont, mathrsfs,amsthm}
\usepackage{float}
\usepackage{mdframed}
\usepackage{mathtools}
\usepackage{appendix}
\usepackage[hidelinks]{hyperref}
\usepackage{graphicx}
\usepackage{cite}
\usepackage{caption}
\usepackage[margin=1in]{geometry}

\allowdisplaybreaks

\title{Sharp Analytical Capacity Upper Bounds\\ for Sticky and Related Channels
}

\author{Mahdi Cheraghchi \and Jo\~ao Ribeiro\thanks{
Department of Computing, Imperial College London, UK. Emails: \{m.cheraghchi, j.lourenco-ribeiro17\}@imperial.ac.uk.}
}

\date{}

\newcommand{\eps}{\epsilon}

\newtheorem{thm}{Theorem}

\newtheorem{lem}[thm]{Lemma}

\newtheorem{coro}[thm]{Corollary}

\newtheorem{remark}[thm]{Remark}

\newcommand{\Ch}{\mathsf{Ch}}
\newcommand{\Ca}{\mathsf{Cap}}
\newcommand{\NB}{\mathsf{NB}}
\newcommand{\Yqd}{Y^{(q)}_\delta}
\newcommand{\Yq}{Y^{(q)}}
\newcommand{\YYq}{\overline{Y}^{(q)}}
\newcommand{\KL}{D_\mathsf{KL}}

\let\originalleft\left
\let\originalright\right
\renewcommand{\left}{\mathopen{}\mathclose\bgroup\originalleft}
\renewcommand{\right}{\aftergroup\egroup\originalright}

\usepackage{url}
\begin{document}

\maketitle

\begin{abstract}
We study natural examples of binary channels with synchronization errors.
These include the duplication channel, which independently outputs
a given bit once or twice, and geometric channels that repeat a given
bit according to a geometric rule, with or without the possibility of
bit deletion. We apply the general framework of Cheraghchi (STOC 2018) 
to obtain sharp analytical upper bounds on the capacity of these channels.
Previously, upper bounds were known via numerical computations involving the
computation of finite approximations of the channels by a computer and
then using the obtained numerical results to upper bound the actual capacity.
While leading to sharp numerical results, further progress on the 
full understanding of the channel capacity inherently remains elusive
using such methods. Our results can be regarded as a major step towards a
complete understanding of the capacity curves. 
Quantitatively, our upper bounds sharply approach, and in some
cases surpass, the bounds that were previously only known by purely
numerical methods. Among our results, we notably give a completely
analytical proof that, when the number of repetitions per bit is geometric (supported on $\{0,1,2,\dots\}$)
with mean growing to infinity, the channel capacity remains substantially bounded
away from $1$.
\end{abstract}


\section{Introduction}\label{sec:intro}

Channels with synchronization errors, such as deletions, replications, and insertions of random bits, have enjoyed significant attention in the past few decades, and more so in recent years. This is due to two reasons: First, our techniques for tackling synchronization errors are still limited relative to memoryless channels. Second, the study of synchronization erros, in addition to being natural, are motivated by practical situations. For example, such channels naturally arise when dealing with DNA-based data storage methods (cf.~\cite{YKGMM15}). 

A well-known example of a channel with synchronization errors is the \emph{binary deletion channel}, which independently removes each input bit from a given bit stream with a certain deletion probability. Determining the exact capacity of the binary deletion channel remains a major challenge in information theory. However, various upper and lower bounds on the channel capacity are known; e.g., \cite{MD06,Dal10,FD10,RD15,Che17}. The behavior of the exact deletion capacity curve is satisfactorily known only for small deletion probabilities \cite{KMS10,KM13}.

Other types of synchronization errors have been considered as well, in particular bit replications caused by timing errors. In this case, each input bit is independently replicated a certain number of times in the output according to a fixed replication probability distribution over the non-negative integers (the distribution defining the repetition rule may have support on the outcome $0$, in which case the given bit is simply deleted).

The difficulty in fully understanding of the capacity of the deletion channel motivates the study of simpler, but still practically relevant, channels where bits can be replicated but never deleted, known as \emph{sticky channels} (for an application-oriented work on sticky channels, see \cite{FMS09}). Although the process defining sticky channels over bits may still not be memoryless, any sticky channel is equivalent to a memoryless channel, and that in principle makes them potentially simpler objects to study than deletion-type channels. This is simply because a sticky channel acts independently on runs of bits, mapping a run of consecutive zeros or ones into one of equal length or longer. As a result, seen as a channel over integer sequences (modeling the run-length encoding of the input bit stream), a sticky channel is equivalent to a memoryless channel over the integers. While the underlying channel over the integers exactly characterizes the original channel, it may also shed light into the understanding of channels that allow deletions, in particular the binary deletion channel. This is because structurally similar channels over the integers (e.g., the binomial channel) arise as natural key intermediate objects in the study of the deletion channel \cite{DMP07,Che17}.

To the best of our knowledge, there is no nontrivial repetition rule for which the capacity of the resulting channel is exactly known (the binary deletion channel being a notorious example). This includes both sticky channels and channels allowing deletions. We consider two natural examples of sticky channels that have been substantially studied in the literature; namely, the \emph{elementary duplication channel} and the \emph{geometric sticky channel}. In the former, a given bit is possibly duplicated with a given probability, and in the latter the number of times each bit is replicated follows a geometric distribution supported on $\{1,2,\dots\}$. Even though deriving an explicit expression for the capacity of these channels is still an outstanding open problem, there are tight numerical lower and upper bounds \cite{DM07,DK07,Mit08,MTL12} and analytical lower bounds \cite{ISW16} on their capacity. Despite the fact that these bounds give us a good idea of the shape of the capacity curve, they do not yield a better conceptual understanding of the capacity, nor do they help us get closer to determining an exact, explicit expression for the channel capacities. 


In this work, we make significant progress towards an analytical characterization of the capacity curve for the elementary duplication and geometric sticky channels. This is achieved by instantiating a general framework developed by one of the authors \cite{Che17} for studying the capacity of channels with synchronization errors by convex programming. 

Roughly speaking, the framework of \cite{Che17} allows one to obtain explicit capacity upper bounds, or even an exact expression for the capacity, by carefully designing explicit distributions that satisfy certain constraints. The quality of the resulting capacity upper bounds generally depend on how tightly the underlying constraints are satisfied by the designed distribution. 
Candidate distributions were derived in \cite{Che17} for the deletion and Poisson-repeat channels.
We derive explicit expressions for candidate distributions corresponding to sticky channels (that actually tightly satisfy all underlying constraints), and subsequently sharp capacity upper bounds via the above-mentioned framework.
We also consider geometric repetition rules with support on zero; i.e., with possibility of bit-deletion,
and capacity upper bounds for such channels. 

Tight numerical upper bounds for the elementary duplication channel were already derived in \cite{Mit08}. Furthermore, there is a line of work studying elementary duplications combined with deletions, insertions, and bit flips \cite{FDE11,RD13,VTR13,RD14}. For such channels, the behavior of the capacity for small duplication probability is well-understood \cite{RA13}. 
The first capacity upper bounds for the geometric sticky channel and channels combining geometric replications and deletions were obtained in \cite{MTL12}.


\subsection{Previous work}

Sticky channels were first introduced and studied by Drinea, Kirsch, and Mitzenmacher \cite{DK07,DM07,Mit08}. Of particular note, Mitzenmacher \cite{Mit08} gave numerical capacity lower bounds for the elementary and geometric sticky channels, along with a tight numerical capacity upper bound for the elementary duplication channel. 

Mercier, Tarokh, and Labeau \cite{MTL12} derive tight numerical capacity upper bounds for the geometric sticky channel. Furthermore, they introduce and study a more general model which combines deletions with geometric replications (and also possibly with insertions and substitutions). More precisely, they consider a more general setting where the channel operates on the input bits in \emph{rounds}. Suppose that the channel is processing bit $x_i$ in round $j$. Then, it either deletes $x_i$ with probability $p_d$ and moves to $x_{i+1}$ in round $j+1$; adds a copy of $x_i$ to the output with probability $p_t$ and stays in $x_i$ in round $j+1$; or adds $x_i$ to the output with probability $1-p_d-p_t$ and moves to $x_{i+1}$ in round $j+1$. They particularly focus on the geometric sticky channel (when $p_d=0$) and on the special case where $p_d=p_t$.

Iyengar, Siegel, and Wolf \cite{ISW16} also study a model similar to that of \cite{MTL12}. They derive analytical expressions for the rates achieved by codebooks generated by Markov chains of several orders in the geometric sticky channel. Moreover, they numerically estimate the rate achieved by codebooks generated by Markov chains of several orders in channels combining geometric replications and deletions (in particular, they focus on the above-mentioned special case $p_d=p_t$) by approximating them by finite-state channels.

It is instructive to consider how the upper bounds in \cite{Mit08,MTL12} are derived. The same technique was also used in \cite{DMP07} to derive capacity upper bounds for the deletion channel. At a high level, given a channel $\Ch$ with synchronization errors (say, an elementary duplication or a geometric sticky channel), one first reduces upper bounding the capacity of $\Ch$ to upper bounding the capacity \emph{per unit cost} of a memoryless channel $\Ch'$. The capacity per unit cost $c$ of a memoryless channel $\Ch'$ is defined as
\[
\sup_X \frac{I(X;Y)}{\mathds{E}[c(X)]},
\]
where $Y$ denotes the channel output distribution induced by input distribution $X$.

For channels with replications only (that do not delete bits), this reduction is straightforward and lossless. This is because one can just consider the operation of the channel on each run of consecutive bits in the input. Upper bounding the capacity per unit cost of memoryless channels with finite input and output alphabets is, on the other hand, possible via the following result of Abdel-Ghaffar's \cite{AG93}:
\begin{thm}[\protect{\cite{AG93}}]\label{thm:ag}
	Consider a discrete memoryless channel $\Ch$ with input alphabet $\mathcal{X}$, output alphabet $\mathcal{Y}$, and output distribution $Y_x$ given input $x$. Let $Y$ be any distribution. Then, the capacity per unit cost $c(x)$ of $\Ch$ is at most
	\[
	\sup_{x\in\mathcal{X}}\frac{\KL(Y_x||Y)}{c(x)},
	\]
	where $\KL(Y_x||Y)$ denotes the Kullback-Leibler (KL) divergence between $Y_x$ and $Y$. Moreover, if $Y$ is a channel output distribution induced by an input distribution $X$ with support $\mathcal{X'}$ and $\KL(Y_x||Y)/c(x)=\lambda$ for all $x\in\mathcal{X'}$, then the capacity of $\Ch$ is exactly $\lambda$ and $X$ is capacity-achieving.
\end{thm}

Analytically designing candidate distributions $Y$ to be used in Theorem~\ref{thm:ag} turns out to be complex even for simple cost functions like $c(x)=x$. Instead, previous works numerically design such distributions by first approximating the capacity and the optimal input distribution for a finite variation of the channel under consideration (e.g., via a variant of the Blahut-Arimoto algorithm). Then, the resulting information is used to design a good candidate distribution by extending either the numerically obtained (finite) input or output distribution with an appropriate tail. Despite this, one may still need to consider genie-aided encoding and decoding to simplify the analysis, inherently leading to sub-optimal results. For example, Mercier, Tarokh, and Labeau \cite{MTL12} consider a modified channel which is designed to be noiseless if its input or output values are large enough (such a modification can only increase the capacity).

Although the above-mentioned approach leads to tight capacity bounds, there are two main drawbacks. First, it does not lead to a better conceptual understanding of the channel. Second, it automatically precludes an exact characterization of the capacity potentially obtainable via Theorem~\ref{thm:ag}.


Recently, one of the authors \cite{Che17} proved a fixed-mean variation of Theorem~\ref{thm:ag}. In this case, if the mean restriction is $\mu\in\mathbb{R}$, one obtains upper bounds for the capacity of channels where the only input distributions allowed are those that induce output distributions with mean $\mu$. While such a statement is technically equivalent to Theorem~\ref{thm:ag} (since both potentially characterize the exact capacity), this subtle change of perspective allows us to actually design good distributions $Y$ purely analytically. Subsequently, this leads to sharp analytical capacity upper bounds which are discussed in more detail in Section \ref{sec:red}.


\subsection{Contributions}

In this work, we study the capacity of three channels: The elementary duplication channel, the geometric sticky channel, and a channel combining geometric replications and deletions, which we call the \emph{geometric deletion channel}.

The elementary duplication channel with duplication probability $p$ receives a string $x_1x_2\dots x_n$ as input and replaces each bit $x_i$ by either one copy of $x_i$ with probability $1-p$, or two copies of $x_i$ with probability $p$.

The geometric sticky channel with replication parameter $p$ receives a string $x_1x_2\dots x_n$ as input and replaces each bit $x_i$ by $D_{i}$ copies of $x_i$, where the $D_{i}$ are i.i.d. and follows a geometric distribution supported on $\{1,2,\dots\}$; i.e.,
\[
\Pr[D_{i}=k]=(1-p)p^{k-1},\quad k=1,2,\dots.
\]

The geometric deletion channel with replication parameter $p$ is similar to the geometric sticky channel, except that the number of times each bit is replicated is distributed according to a geometric distribution supported on $\{0,1,2,\dots\}$. That is, in this case we have 
\[
\Pr[D_{i}=k]=(1-p)p^k,\quad k=0,1,2,\dots.
\]
In the more general model for geometric replications and deletions introduced in \cite{MTL12}, the $0$-geometric channel corresponds to the case where the deletion probability $p_d$ and the duplication probability $p_t$ satisfy $p_d=1-p_t$.

Our contributions are threefold, summarized below.

\paragraph{Sticky channels}
We derive analytical capacity upper bounds for the elementary duplication and geometric sticky channels which are tight over a large range of parameters. Furthermore, the bounds are supremums of analytic, uni-variate, concave functions over $(0,1)$, and so can be easily computed. Our results can be interpreted as the first evidence that determining the exact capacity of some sticky channels may be within reach. In fact, our upper bounds are obtained by first designing distributions which satisfy the constraint in the fixed-mean analogue of Abdel-Ghaffar's result in \cite{Che17} \emph{with equality}. If these distributions are also shown to be valid channel output distributions, then this implies that we have obtained an exact expression for the capacity of the underlying channel. While this turns out to not be the case, it may be possible to adapt our techniques to achieve this. 

The bounds we obtain are very sharp when the duplication probability is not too large. For example, the analytical capacity upper bound for the geometric sticky channel is within $10^{-5}$ of the numerical upper bound given in \cite{MTL12} for $p\leq 0.5$. Moreover, we improve upon the known numerical upper bounds for both the geometric sticky and elementary duplication channels for some values of the duplication probability.

\paragraph{The geometric deletion channel}
We design distributions and derive improved capacity upper bounds for the geometric deletion channel. These improvements are obtained by combining the distribution design techniques from \cite{Che17} with a simple refinement.

\paragraph{The large replication regime}
Finally, we give a fully analytical proof that, rather counter-intuitively, the capacity of the geometric deletion channel is at most $0.73$ bits/channel use (thus significantly bounded away from $1$) when the replication parameter $p$ approaches 1 (i.e., the expected number of replications grows to infinity, or, equivalently, the deletion probability approaches $0$). This stands in stark contrast to the deletion and Poisson-repeat channels, whose capacities converge to $1$ when the deletion probability approaches $0$ (see Appendix~\ref{app:poisson} for a proof of this fact for the Poisson-repeat channel). Note that the Poisson-repeat channel case shows that there are channels defined by repetition rules with full support over the non-negative integers whose capacity approaches $1$ when the expected number of replications grows to infinity. As a result, it is not clear at first sight whether the capacity of the geometric deletion channel approaches $1$ or not in this setting.

\subsection{Notation}

We denote the natural logarithm by $\log$. We will be dealing solely with discrete random variables, which are denoted by uppercase letters such as $X$, $Y$, and $Z$. The expected value of $X$ is denoted by $\mathds{E}[X]$, and in general we write $X(x)$ for the probability that $X$ takes on value $x$. We denote the Kullback-Leibler divergence between $X$ and $Y$ by $\KL(X||Y)$. We denote the Shannon entropy of $X$ by $H(X)$ and the binary entropy function by $h$.

\subsection{Organization}

The rest of the article is organized as follows. In Section~\ref{sec:red}, we describe and discuss the general framework developed in \cite{Che17} for studying the capacity of channels with synchronization. In Section~\ref{sec:geomsticky}, we study the geometric sticky. We derive analytical capacity upper bounds for the geometric sticky channel in Section~\ref{sec:zerogapsticky}, and compare them to the known numerical bounds in Section~\ref{sec:boundsticky}. In Section~\ref{sec:dupl}, we study the elementary duplication channel.  We derive analytical capacity upper bounds for the elementary duplication in Section~\ref{sec:zerogapdupl}, and compare them to the known numerical bounds in Section~\ref{sec:bounddupl}. We study the geometric deletion channel in Section~\ref{sec:geomdel}. General analytical bounds are derived in Sections~\ref{sec:firstbound}~and~\ref{sec:trunc}. Improved bounds are obtained by considering the refinement described in Section~\ref{sec:fixmass0}. These bounds are compared to the known ones in Section~\ref{sec:compbounds}. The fully analytical capacity upper bound for large replication parameter is derived in Section~\ref{sec:boundsmalld}.

\section{Reduction to a memoryless channel and the convex duality framework}\label{sec:red}

In this section, we introduce the general reduction of repeat channels to memoryless channels and the convex duality framework developed in \cite{Che17}.

Consider a random variable $D$ supported on the non-negative integers. We denote by $\Ch(D)$ the \emph{repeat channel with replication distribution $D$}, which works as follows: For each input bit $x_i\in\{0,1\}$, the channel replaces $x_i$ with $D_i$ copies of $x_i$, where the $D_i$ are i.i.d.\ and distributed according to $D$. 

It will be useful to define the concept of a \emph{mean-limited channel}. Given a channel $\Ch$ with input and output alphabets contained in $\mathbb{R}$, we denote by $\Ch_\mu$ the channel with the same channel law as $\Ch$, but where the input distributions are restricted to only those that induce output distributions $Y$ satisfying $\mathds{E}[Y]=\mu$. Then, $\Ch_\mu$ is the mean-limited version of $\Ch$.

The following theorem relates the capacity of $\Ch(D)$, which we denote by $\Ca(D)$, with the capacity of an associated mean-limited memoryless channel.
\begin{thm}[\cite{Che17}]\label{thm:red}
	Fix a distribution $D$ over the non-negative integers, and let $\overline{D}$ denote $D$ conditioned on the event $D\neq 0$. Let $\Ch'(D)$ denote the channel which on input $1+x$ for $x\in\{0,1,\dots\}$ outputs $\overline{D}+\sum_{i=1}^x D_i$, where $\overline{D}$ and the $D_i$ are independent, and furthermore the $D_i$ are all distributed according to $D$. Let $\Ca'_\mu(D)$ denote the capacity of $\Ch'_\mu(D)$. Then,
	\begin{equation}
	\Ca(D)\leq \sup_{\mu\geq \overline{\lambda}}\frac{\Ca'_\mu(D)}{1/p + (\mu-\overline{\lambda})/\lambda},
	\end{equation}
	where $\lambda=\mathds{E}[D]$, $\overline{\lambda}=\mathds{E}[\overline{D}]$, and $p=1-D(0)$.
\end{thm}

In the case of sticky channels, where $D(0)=0$ and hence $\overline{D}=D$, the reduction in Theorem~\ref{thm:red} does not incur any loss. As a result, we obtain an exact characterization of $\Ca(D)$ in terms of the capacity of a memoryless channel. Using that for such a channel we have $p=1-D(0)=1$ and $\lambda=\overline{\lambda}$ leads to
\begin{equation}\label{eq:stickychar}
\Ca(D)=\lambda\sup_{\mu\geq \lambda}\frac{\Ca'_\mu(D)}{\mu}.
\end{equation}


It remains now to upper bound $\Ca'_{\mu}(D)$ for general $\mu$. This can be achieved via the following theorem, which can be interpreted as a mean-limited version of Abdel-Ghaffar's duality-based characterization of the capacity per unit cost \cite{AG93}.
\begin{thm}[\cite{Che17}]\label{thm:duality}
	Fix a channel $\Ch$ with input and output alphabets $\mathcal{X},\mathcal{Y}\subseteq\mathbb{R}$, respectively, and let $Y_x$ denote the output distribution of $\Ch$ when $x$ is input into the channel. If there exist parameters $a,b\in\mathbb{R}$ and a distribution $Y$ such that
	\begin{equation}
	D_{\mathsf{KL}}(Y_x||Y)\leq a \mathds{E}[Y_x]+b
	\end{equation}
	for all $x\in\mathcal{X}$, then the capacity of the mean-limited channel $\Ch_\mu$ is at most
	\begin{equation*}
	a\mu+b.
	\end{equation*}
	Moreover, if there is an input distribution $X$ with support $\mathcal{X'}$ that induces $Y$ as the channel output distribution, $\mathds{E}[Y]=\mu$, and
	\begin{equation}
	D_{\mathsf{KL}}(Y_x||Y)= a \mathds{E}[Y_x]+b
	\end{equation}
	for all $x\in\mathcal{X'}$, then the capacity of $\Ch_\mu$ is exactly $a\mu+b$ and $X$ is a capacity-achieving distribution.
\end{thm}

An important concept when dealing with Theorem~\ref{thm:duality} is the \emph{KL-gap} of a distribution $Y$, which we proceed to explain. Fix a channel $\Ch$ with input alphabet $\mathcal{X}$, let $Y_x$ be the output distribution given input $x$, and suppose some distribution $Y$ satisfies
\[
D_{\mathsf{KL}}(Y_x||Y)\leq a\mathds{E}[Y_x]+b
\]
for all $x\in\mathcal{X}$. Then, the KL-gap of $Y$ with respect to $a$, $b$, and $\Ch$ is defined (as a function of $x$) as
\[
a\mathds{E}[Y_x]+b-\KL(Y_x||Y).
\]
A good goal when designing a distribution $Y$ for Theorem~\ref{thm:duality} is to minimize the KL-gap as much as possible, for two reasons: First, from experience it appears to lead to overall better capacity upper bounds. Second, designing distributions with zero KL-gap is a first step towards determining the channel capacity exactly, the remaining step being that these distributions should also be realizable as channel output distributions. This is the philosophy behind the design techniques developed in \cite{Che17}, although it was not possible to construct distributions with zero KL-gap everywhere.

\section{The geometric sticky channel}\label{sec:geomsticky}

In this section, we study the capacity of the geometric sticky channel. As discussed before, the current known bounds require significant computational power, and their derivation makes use of a variant of the Blahut-Arimoto algorithm to obtain good distributions to be used in conjunction with Theorem~\ref{thm:ag}. This means that there is no analytical method behind the design of these distributions.

We make progress towards an analytical understanding of the capacity by designing a family of distributions with zero KL-gap for the memoryless channel associated to the geometric sticky channel. Furthermore, for every $\mu$ there is a distribution $Y$ in this family which satisfies $\mathds{E}[Y]=\mu$. This is a significant step towards obtaining an exact analytical expression for the capacity of the geometric sticky channel, since Theorem~\ref{thm:duality} states that if such distributions are also valid channel output distributions, then we have determined the capacity exactly.

The geometric sticky channel independently replicates each input bit according to a distribution $D_1$ satisfying
\[
D_1(y)=(1-p) p^{y-1},\quad y=1,2,\dots
\]
for some $p\in[0,1)$ which we call the \emph{replication parameter}, i.e., $D_1$ follows a geometric distribution with success probability $1-p$ supported on $\{1,2,\dots\}$. In order to use \eqref{eq:stickychar} combined with Theorem~\ref{thm:duality}, we need to understand the channel $\Ch'(D_1)$ which on input $x\in\{1,2,\dots\}$ outputs
\[
Y_x=\sum_{i=1}^x D_{1i},
\]
where the $D_{1i}$ are i.i.d.\ according to $D_1$. This is because for the geometric sticky channel we have $D_1=\overline{D_1}$. For any input $x\in\{1,2,\dots\}$, the output channel distribution has a nice form. More precisely, if $Y_x$ denotes the channel output distribution given input $x$, then
\begin{equation}\label{eq:yxgeom}
Y_x=x+\mathsf{NB}_{x,p},
\end{equation}
where $\mathsf{NB}_{x,p}$ denotes the negative binomial distribution with $x$ successes and success probability $p$, which satisfies
\begin{equation}\label{eq:NB}
\mathsf{NB}_{x,p}(y)=\binom{y+x-1}{y}(1-p)^x p^y,\quad y=0,1,\dots
\end{equation}
That \eqref{eq:yxgeom} holds follows easily from the fact that $D_1=1+D_0$, where $D_0$ follows a geometric distribution with success probability $1-p$ supported on $\{0,1,2,\dots\}$, i.e.,
\[
D_0(y)=(1-p)p^y,\quad y=0,1,\dots
\]
and that $\mathsf{NB}_{x,p}=\sum_{i=1}^x D_{0i}$, where the $D_{0i}$ are i.i.d.\ according to $D_0$.

As a consequence, it follows easily from \eqref{eq:yxgeom} and \eqref{eq:NB} that
\begin{equation}\label{eq:outgeom}
Y_x(y)=\binom{y-1}{x-1}(1-p)^x p^{y-x},\quad y=x,x+1,\dots
\end{equation}
for all $x\geq 1$.

\subsection{A distribution with zero KL-gap everywhere}\label{sec:zerogapsticky}

In this section, we show how to design distributions for Theorem~\ref{thm:duality} with zero KL-gap for the geometric sticky channel.

We begin by noting that $\KL(Y_x||Y)$ can be rewritten as
\[
\KL(Y_x||Y)=-H(Y_x)-\sum_y Y_x(y)\log Y(y).
\]
Then, recalling \eqref{eq:outgeom} and noting that $\mathds{E}[Y_x]=\frac{x}{1-p}$, we have
\begin{align}\label{eq:KLgeomsticky}
-H(Y_x)&=\mathds{E}\left[\log\binom{Y_x-1}{x-1}\right]+x\log(1-p)+(\mathds{E}[Y_x]-x)\log p\nonumber\\
&=\mathds{E}\left[\log\binom{Y_x-1}{x-1}\right]-\mathds{E}[Y_x]h(p)\nonumber\\
&= \mathds{E}[\log(Y_x-1)!]-\mathds{E}[\log(Y_x-x)!]-\log(x-1)!-\mathds{E}[Y_x]h(p).
\end{align}

Our goal is to design a family of distributions $Y$ such that $\KL(Y_x||Y)$ is an affine function of $\mathds{E}[Y_x]$. Given $q\in(0,1)$, consider the distribution $\Yq$ with general form
\[
\Yq(y)=y_0q^y\exp(g(y)-yh(p)),\quad y=1,2,\dots
\]
where $y_0$ is the normalizing factor and $g$ is a function to be defined. Then, using \eqref{eq:KLgeomsticky},
\begin{align}\label{eq:KLYq}
\KL(Y_x||\Yq)&=-H(Y_x)-\sum_y Y_x(y)\log \Yq(y)\nonumber\\
&=\mathds{E}\left[\log\binom{Y_x-1}{x-1}\right]-\mathds{E}[Y_x]h(p) -\log y_0 -\mathds{E}[Y_x]\log q\nonumber - \mathds{E}[g(Y_x)]+\mathds{E}[Y_x]h(p)\nonumber\\
&=-\log y_0+\mathds{E}[Y_x]\log q +\mathds{E}[\log(Y_x-1)!]-\mathds{E}[\log(Y_x-x)!] -\log(x-1)!-\mathds{E}[g(Y_x)].
\end{align}

Taking into account \eqref{eq:KLYq}, ideally we would like to have
\begin{equation}\label{eq:ggoal}
\mathds{E}[g(Y_x)]=\mathds{E}[\log(Y_x-1)!]-\mathds{E}[\log(Y_x-x)!]-\log(x-1)!
\end{equation}
for all $x\geq 1$, so that
\[
\KL(Y_x||\Yq)=-\log y_0 - \mathds{E}[Y_x]\log q.
\]

We will proceed to design such a function $g$. Before we begin, we first state some lemmas that will be useful later on. The following result gives an integral representation of the log gamma function.
\begin{lem}[\cite{Che17a}]\label{lem:intlog}
	We have
	\[
	\log \Gamma(1+z)=\int_0^1 \frac{1-tz-(1-t)^z}{t\log(1-t)}dt
	\]
	for all $z\geq 0$.
\end{lem}

We will require a version of Fubini's theorem specialized for the counting measure on $\mathbb{N}$ and the Lebesgue measure on $[0,1]$.
\begin{lem}\label{lem:fubini}
	Let $(f_n)_{n\in\mathbb{N}}$ be a family of continuous functions $f_n:[0,1]\to \mathbb{R}$, and suppose that either
	\[
	\int_0^1 \sum_{n=1}^\infty |f_n(t)| dt<\infty
	\]
	or
	\[
	\sum_{n=1}^\infty \int_0^1 |f_n(t)| dt<\infty.
	\]
	Then,
	\[
	\int_0^1 \sum_{n=1}^\infty f_n(t) dt=\sum_{n=1}^\infty \int_0^1 f_n(t) dt.
	\]
\end{lem}

Making use of Lemmas \ref{lem:intlog} and \ref{lem:fubini}, and of the facts that $\mathds{E}[Y_x]=\frac{x}{1-p}$ and that the probability generating function of $Y_x$ is
\begin{equation}\label{eq:pgfsticky}
\left(\frac{z(1-p)}{1-pz}\right)^x,
\end{equation}
we have
\[
\log(x-1)!=\int_0^1 \frac{1+t-t x-(1-t)^{x-1}}{t\log(1-t)}dt
\]
and\footnote{We can justify the switching of the integral and expected value in \eqref{eq:explog} via Lemma~\ref{lem:fubini} by noting that the function inside the integral in Lemma~\ref{lem:intlog} can be extended by continuity to $[0,1]$, is non-negative for all $y\geq x$, and that the left-hand side of \eqref{eq:explog} is finite.}
\begin{align}\label{eq:explog}
\mathds{E}[\log(Y_x-x)!]&=\mathds{E}\left[\int_0^1 \frac{1-t(Y_x-x)-(1-t)^{Y_x-x}}{t \log(1-t)}dt\right]\nonumber\\
&=\int_0^1 \frac{1-\frac{txp}{1-p}-\left(\frac{1-p}{1-p(1-t)}\right)^x}{t \log(1-t)}dt.
\end{align}
Consider the functions
\begin{align}
f_1(y,t)&=\frac{1+t-ty(1-p)-\left(\frac{1-t}{1-pt}\right)^y/(1-t)}{t\log(1-t)}\label{eq:f1sticky}\\
f_2(y,t)&=\frac{1-typ-\left(\frac{1}{1+pt}\right)^y}{t\log(1-t)}.\label{eq:f2sticky}
\end{align}
Recalling \eqref{eq:pgfsticky}, observe that
\begin{align}\label{eq:expfi}
\mathds{E}[f_1(Y_x,t)]&=\frac{1+t-t x-(1-t)^{x-1}}{t\log(1-t)},\nonumber\\
\mathds{E}[f_2(Y_x,t)]&=\frac{1-\frac{txp}{1-p}-\left(\frac{1-p}{1-p(1-t)}\right)^x}{t \log(1-t)}.
\end{align}
With \eqref{eq:ggoal} in view, we set
\begin{align}\label{eq:lambdasticky}
\Lambda_1(y)&=\int_0^1 f_1(y,t)dt,\nonumber\\
\Lambda_2(y)&=\int_0^1 f_2(y,t)dt
\end{align}
with $f_1$ and $f_2$ defined as in \eqref{eq:f1sticky} and \eqref{eq:f2sticky}, respectively. Taking into account \eqref{eq:expfi}, we show the following.
\begin{lem}\label{lem:switch}
	We have
	\begin{align*}
	\mathds{E}[\Lambda_1(Y_x)]&=\int_0^1 \mathds{E}[f_1(Y_x,t)]dt=\log(x-1)!,\\
	\mathds{E}[\Lambda_2(Y_x)]&=\int_0^1 \mathds{E}[f_2(Y_x,t)]dt=\mathds{E}[\log(Y_x-x)!].
	\end{align*}
\end{lem}
\begin{proof}
	The only problem lies with the first equality (the second equality follows directly from \eqref{eq:expfi}). We start by showing that this equality holds for $\Lambda_1$. All we need to do is see that the conditions in Lemma~\ref{lem:fubini} are satisfied. 
	
	First, it is easy to see that $f_1(y,\cdot)$ is continuous in $(0,1)$, and that
	\begin{equation}\label{eq:lim1}
	\lim_{t\to 1}f_1(y,t)=0
	\end{equation}
	and
	\begin{equation}\label{eq:lim0}
	\lim_{t\to 0}f_1(y,t)=1+\frac{y(1-p)(y(1-p)-3-p)}{2}
	\end{equation}
	for all $y\geq 1$. This means that $f_1(y,\cdot)$ can be extended by continuity to $[0,1]$ (this does not change the integral). From here onwards we work with this extension. By Lemma~\ref{lem:fubini}, we only need to show that
	\begin{equation}\label{eq:absint}
	\int_0^1 \mathds{E}[|f_1(Y_x,t)|]dt<\infty.
	\end{equation}
	
	We begin by showing that $f_1(y,t)\geq 0$ for all $t\in[0,1]$ if $y$ is large enough. Recalling \eqref{eq:f1sticky}, the numerator of $f_1(y,t)$ is
	\begin{equation}\label{eq:num}
	h(y,t):=1+t-ty(1-p)-\left(\frac{1-t}{1-pt}\right)^y/(1-t).
	\end{equation}
	We show that $h(y,t)\leq 0$ for all $t\in[0,1]$ if $y$ is large enough. This gives the desired result since the denominator of $f_1(y,t)$ is negative for all $t\in(0,1)$. The first and second derivatives with respect to $t$ of $h(y,t)$ are
	\begin{align*}\label{eq:der}
	\frac{\partial h}{\partial t}(y,t)&=1-y(1-p)+\frac{(y(1-p) + pt-1) \left(\frac{1-t}{1-pt}\right)^{y+1}}{(1-t)^3},\nonumber\\
	\frac{\partial^2 h}{\partial t^2}(y,t)&=\left(\frac{1-t}{1-pt}\right)^y\bigg(\frac{(1-p)(3+p(1-4t)y)-(1-p)^2 y^2+2(1-pt)^2}{(1-t)^3(1-pt)^2}\bigg).
	\end{align*}
	For fixed $p$, we can set $y^*$ large enough (and independent of $t$) so that
	\begin{align*}
	&(1-p)(3+p(1-4t)y)-(1-p)^2 y^2-2(1-pt)^2\leq 4(1-p)y-(1-p)^2 y^2<0
	\end{align*}
	for all $t\in[0,1]$ and $y\geq y^*$, which implies that $\frac{\partial^2 h}{\partial t^2}(y,t)<0$ for all $t\in[0,1]$. As a consequence, it follows that $h'(y,t)$ is decreasing in $t$ for $y\geq y^*$. Combining this with the fact that $\frac{\partial h}{\partial t}(p,y,0)=0$, we conclude that $\frac{\partial h}{\partial t}(y,t)\leq 0$ for all $t\in[0,1]$, provided that $y\geq y^*$. Finally, this implies that $h(y,t)\leq 0$ holds for $y\geq y^*$, since $h(y,0)=0$.
	
	
	Consequently, we have
	\begin{align*}
	\eqref{eq:absint}&\leq \int_0^1 \left(\mathds{E}[f_1(Y_x,t)]+2\sum_{y=1}^{y^*}Y_x(y)|f_1(y,t)|\right)dt\\
	&= \log(x-1)!+C_{x,p}<\infty,
	\end{align*}
	where
	\[
	C_{x,p}=2\int_0^1 \sum_{y=1}^{y^*}Y_x(y)|f_1(y,t)| dt
	\]
	is a finite constant depending only on $x$ and $p$, since $f_1(y,\cdot)$ is continuous in $[0,1]$ for all $y\geq 1$, and therefore bounded as well. This means that Lemma~\ref{lem:fubini} can be applied, which leads to the desired equality.
	
	The argument for $\Lambda_2$ follows in an analogous, but simpler, way. In fact, recalling \eqref{eq:f2sticky}, the numerator of $f_2(y,t)$ is
	\begin{equation}\label{eq:num2}
	h_2(y,t)=1-typ-\left(\frac{1}{1+pt}\right)^y,
	\end{equation}
	and its derivative with respect to $t$ is
	\begin{equation}\label{eq:dernum2}
	\frac{\partial h_2}{\partial t}(y,t)=yp\left(\left(\frac{1}{1+pt}\right)^{1+y}-1\right).
	\end{equation}
	It is clear that $\frac{\partial h_2}{\partial t}(y,t)< 0$ for $t\in(0,1)$ and $y\geq 1$, which implies that $h_2(y,t)$ is decreasing in $t$ for fixed $p$ and $y\geq 1$. Combining this with the fact that $h_2(y,0)=0$ for all $y\geq 1$ yields that $f_2(y,t)\geq 0$ for all $t\in(0,1)$ and $y\geq 1$. As before, it is easy to see that $f_2(y,\cdot)$ can be extended by continuity to $[0,1]$. This means we can apply Lemma~\ref{lem:fubini} and obtain the desired result.
\end{proof}

Consider the distribution $\Yq$ defined by the choice of $g$
\begin{equation}\label{def:g}
g(y)=\log(y-1)!-\Lambda_1(y)-\Lambda_2(y).
\end{equation}
By Lemma~\ref{lem:switch}, it follows that $g$ satisfies \eqref{eq:ggoal}, and so, recalling \eqref{eq:KLYq}, we have 
\begin{equation*}\label{eq:finalKLlinegeom}
\KL(Y_x||\Yq)=-\log y_0-\mathds{E}[Y_x]\log q
\end{equation*}
provided that $\Yq$ is a valid distribution. In order to wrap everything up, it remains to show this fact, i.e., that
\[
1/y_0=\sum_{y=1}^\infty q^y \exp(g(y)-y h(p))<\infty
\]
if $q\in(0,1)$, and thus $\Yq$ can be normalized so that $\sum_{i=1}^\infty \Yq(y)=1$. The following lemma implies this by showing that $Y(y)/y_0=\Theta(q^y/\sqrt{y})$.
\begin{lem}\label{lem:asymp}
	We have
	\[
	yh(p)-g(y)=\frac{1}{2}\log y+O(1).
	\]
	for $y\geq 1$ and $p\in[0,1)$.
\end{lem}
\begin{proof}
	It holds that
	\begin{align}
	\Lambda_1(y)&=\log \Gamma(y(1-p))+O(1),\nonumber\\
	\Lambda_2(y)&=\log \Gamma(1+yp)+O(1).\label{eq:asymplambda}
	\end{align}
	We only show the first equality; the second one follows in an analogous manner (we discuss the deviations briefly). Using Lemma~\ref{lem:intlog}, we have
	\begin{equation*}
	\log \Gamma(y(1-p))=\int_0^1 \frac{1+t-ty(1-p)-(1-t)^{y(1-p)-1}}{t\log(1-t)}dt.
	\end{equation*}
	Recalling the definition of $\Lambda_1(y)$, it follows that
	\begin{equation*}
	\Lambda_1(y)-\log \Gamma(y(1-p))=\int_0^1 \frac{\left(\frac{1-t}{1-pt}\right)^y -(1-t)^{y(1-p)}}{t(1-t)\log(1-t)}dt.
	\end{equation*}
	First, observe that for any fixed constant $0<\epsilon<1$ we have
	\[
	\int_\epsilon^1 \frac{\left(\frac{1-t}{1-pt}\right)^y -(1-t)^{y(1-p)}}{t(1-t)\log(1-t)}dt\to 0
	\]
	when $y\to\infty$. As a result, it suffices to show that
	\begin{equation}\label{eq:asympint}
	\int_0^\epsilon \frac{(1-t)^{y(1-p)}-\left(\frac{1-t}{1-pt}\right)^y }{-t\log(1-t)}dt=O(1)
	\end{equation}
	for some $0<\epsilon<1$, since $1-\eps\leq 1-t\leq 1$ for $t\leq \eps$.
	
	Define $h_1(t)=\log\frac{1-t}{1-pt}$, $h_2(t)=(1-p)\log(1-t)$, and $\delta(t)=h_2(t)-h_1(t)$. It is easy to see that $\delta(t)>0$ for all $t\in(0,1)$. We can rewrite the left-hand side of \eqref{eq:asympint} as
	\begin{equation}\label{eq:epsint}
	\int_0^\eps \frac{e^{h_2(t)y}-e^{h_1(t)y}}{-t\log(1-t)}dt.
	\end{equation}
	Then, we have
	\begin{align*}
	\eqref{eq:epsint}&<\int_0^\epsilon \frac{\delta(t)ye^{h_2(t)y}}{-t\log(1-t)}dt\\
	&<c\int_0^\epsilon ye^{h_2(t)y}dt\\
	&<c\int_0^\epsilon ye^{-(1-p)ty}dt\\
	&<c\int_0^\infty ye^{-(1-p)ty}dt=\frac{c}{1-p},\\
	\end{align*}
	for some constant $c>0$. The first inequality follows from the fact that
	\[
	e^{by}-e^{ay} < (b-a)ye^{by}
	\]
	if $y\geq 0$ and $b>a$ (recall that $h_2(t)>h_1(t)$ for all $t>0$). The second inequality follows because $\delta(t)\leq c t^2$ for some constant $c>0$ depending on $\epsilon<1$ (this can be seen by computing the Taylor expansion of $\delta(t)$ around $t=0$), and since $\log(1-t)<-t$. The third inequality stems again from the fact that $\log(1-t)<-t$. The fourth inequality holds because the function inside the integral is positive. It follows that \eqref{eq:asympint} holds, as desired.
	
	We make a brief comment regarding the argument for $\Lambda_2$. We follow the same reasoning as for $\Lambda_1$, but with $h_2(t)=-\log(1+pt)$ and $h_1(t)=p\log(1-t)$. We reduce the problem to showing that $\int_0^\epsilon y e^{-\frac{pt}{1+pt}y}dt$ is bounded by a constant for some $\epsilon<1$. This is can be seen to be true by noting that
	\[
		\int_0^\epsilon y e^{-\frac{pt}{1+pt}y}dt < \int_0^1 y e^{-\frac{pt}{2}y}dt<\int_0^\infty y e^{-\frac{pt}{2}y}dt=2/p=O(1).
	\]
	
	To finalize the derivation, we make use of the asymptotic expansion for the log-gamma function \cite[Sections 6.1.41 and 6.1.42]{AS65}
	\begin{equation}\label{eq:exploggamma}
	\log \Gamma(z)=(z-1/2)\log z - z + O(1)
	\end{equation}
	when $z\geq c>0$ for some constant $c>0$ (the hidden constant in \eqref{eq:exploggamma} depends on $c$). Taking into account \eqref{def:g}, we can apply \eqref{eq:exploggamma} to $\log(y-1)!=\log\Gamma(y)$, $\Lambda_1(y)$, and $\Lambda_2(y)$ (by recalling \eqref{eq:asymplambda}) to obtain
	\begin{align*}
	yh(p)-g(y)&=\frac{1}{2}\log y+yp\log\left(\frac{1+yp}{yp}\right)+O(1)\\
	&=\frac{1}{2}\log y+O(1),
	\end{align*}
	
	which concludes the proof.
\end{proof}

From the results of this section, it follows that $\Yq$ is a valid distribution and that
\begin{equation*}
\KL(Y_x||\Yq)=-\log y_0-\mathds{E}[Y_x]\log q
\end{equation*}
for all $x\geq 1$. Therefore, $\Yq$ achieves zero KL-gap for all $x\geq 1$ and $q\in(0,1)$.

Using Theorem~\ref{thm:duality}, we conclude that
\begin{equation}\label{eq:boundmugeom}
\Ca'_\mu(D_1)\leq \inf_{q\in(0,1)}(-\log y_0 -\mu \log q)
\end{equation}
for all $\mu\geq 1$.

Finally, we point out that Lemma~\ref{lem:asymp} implies that, given any $\mu\geq 1$, there is $q\in(0,1)$ such that $\mathds{E}[\Yq]=\mu$. This will lead to easier to compute, but still tight, capacity upper bounds in Section \ref{sec:boundsticky}.

\subsection{Bounds for the geometric sticky channel}\label{sec:boundsticky}

In this section, we derive an analytical capacity upper bound for the geometric sticky channel by combining the family of distributions $\Yq$ designed in Section \ref{sec:zerogapsticky} with Theorem~\ref{thm:duality} and \eqref{eq:stickychar}, and compare it to the known numerical lower and upper bounds from \cite{Mit08,MTL12}.

In order to apply \eqref{eq:stickychar}, note that $\lambda=\mathds{E}[D_1]=1/(1-p)$. The bound follows by combining \eqref{eq:stickychar} with Theorem~\ref{thm:duality} and \eqref{eq:boundmugeom}.
\begin{coro}
	For every $p\in (0,1)$, we have
	\begin{align}
	\Ca(D_0)&\leq\sup_{\mu\geq 1/(1-p)} \frac{\inf_{q\in(0,1)}(-\log y_0-\mu\log q)}{\mu (1-p)}\\
	&\leq \sup_{q\in(0,1):\atop\mathds{E}[\Yq]\geq 1/(1-p)}\frac{-\log y_0-\mathds{E}[\Yq]\log q}{\mathds{E}[\Yq](1-p)},\label{eq:boundgeomsticky}
	\end{align}
	where
	\begin{align*}
	1/y_0&=\sum_{y=1}^\infty q^y e^{\log(y-1)!-\Lambda_1(y)-\Lambda_2(y)-yh(p)}<\infty,\\
	\mathds{E}[\Yq]&=\sum_{y=1}^\infty y_0yq^y e^{\log(y-1)!-\Lambda_1(y)-\Lambda_2(y)-yh(p)},
	\end{align*}
	with $\Lambda_1$ and $\Lambda_2$ defined as in \eqref{eq:lambdasticky}.
\end{coro}


We remark that \eqref{eq:boundgeomsticky} is obtained by choosing, for each $\mu\geq 1/(1-p)\geq 1$, the value of $q\in (0,1)$ that satisfies $\mathds{E}[\Yq]=\mu$. Lemma~\ref{lem:asymp} ensures that such $q$ always exists for every $\mu\geq 1$.

Table \ref{table:compsticky} compares the results obtained via the analytical capacity upper bound \eqref{eq:boundgeomsticky} with the numerical bounds from \cite{MTL12}. The lower bound from \cite{MTL12} is obtained by numerically optimizing the achievable rate of codebooks generated by 4th order Markov chains. The upper bound from \cite{MTL12} is obtained via a combination of Theorem~\ref{thm:ag} and a variant of the Blahut-Arimoto algorithm. In contrast, our bound has an analytical expression and is derived without computer assistance other than maximizing a concave function over $(0,1)$. Figure~\ref{fig:geomstickycomp} plots the numerical capacity upper bound from \cite{MTL12} and the analytical upper bound \eqref{eq:boundgeomsticky}.

It is also instructive to analyze the behavior of the function inside the supremum in \eqref{eq:boundgeomsticky}. Figure~\ref{fig:geomstickycurves} showcases the behavior of this inner function for some values of $p$. As can be observed, the inner function is concave whenever it is non-negative.

We see that, for $p\leq 0.4$, we are off the numerical upper bound by less than $10^{-6}$. In fact, the error for $p\leq 0.5$ is still less than $10^{-5}$. This shows that our analytical bound is extremely tight whenever $p\leq 0.5$. We also improve over the numerical upper bound for $p=0.15$. However, the bound degrades when $p$ is large; When $p=0.85$, the difference between the analytical and numerical bound is of approximately $0.0117$. For $p\geq 0.9$, the bound increases.

\begin{table}[h]
	\caption{Comparison between the numerical capacity bounds for the geometric sticky channel from \cite{MTL12} and the upper bound \eqref{eq:boundgeomsticky} in bits/channel use.}\label{table:compsticky}
	\centering
	\begin{tabular}{ |c|c|c|c| }
		\hline
		$p$ & Lower bound \cite{MTL12} & Upper bound \cite{MTL12} & Upper bound \eqref{eq:boundgeomsticky}\\ 
		\hline
		0.05 & 0.814457 & 0.814464 &0.814464\\ 
		\hline
		0.10 & 0.714096 & 0.714114 &0.714114\\ 
		\hline
		0.15 & 0.640901 & 0.643267 &0.640930\\ 
		\hline
		0.20 & 0.583575 & 0.583611 &0.583611\\ 
		\hline
		0.25 & 0.537038 & 0.537076 &0.537076\\ 
		\hline
		0.30 & 0.498427 &  0.498463 & 0.498463\\ 
		\hline
		0.35 & 0.465925 & 0.465957 &0.465957\\ 
		\hline
		0.40 & 0.438291 & 0.438318 &0.438318\\ 
		\hline
		0.45 & 0.414637 & 0.414659 &0.414660\\ 
		\hline
		0.50 & 0.394311 & 0.394331 &0.394333\\ 
		\hline
		0.55 & 0.376821 & 0.376849 &0.376855\\ 
		\hline
		0.60 & 0.361775 & 0.361794 &0.361875\\ 
		\hline
		0.65 & 0.348491 & 0.348575 &0.349152\\ 
		\hline
		0.70 & 0.336593 & 0.336946 &0.338551\\ 
		\hline
		0.75 & 0.325900 & 0.326678 &0.330062\\ 
		\hline
		0.80 & 0.316257 & 0.317317 &0.323856\\ 
		\hline
		0.85 & 0.307560 & 0.308767 &0.320448\\ 
		\hline
		0.90 & 0.299601 & 0.300952 &0.321210\\ 
		\hline
		0.95 & 0.292373 & 0.293788 &0.330824\\ 
		\hline
		0.99 & 0.287036 & 0.288476 &0.368459\\ 
		\hline
	\end{tabular}
\end{table}

\begin{figure}
	\centering
	\includegraphics[width=0.7\textwidth]{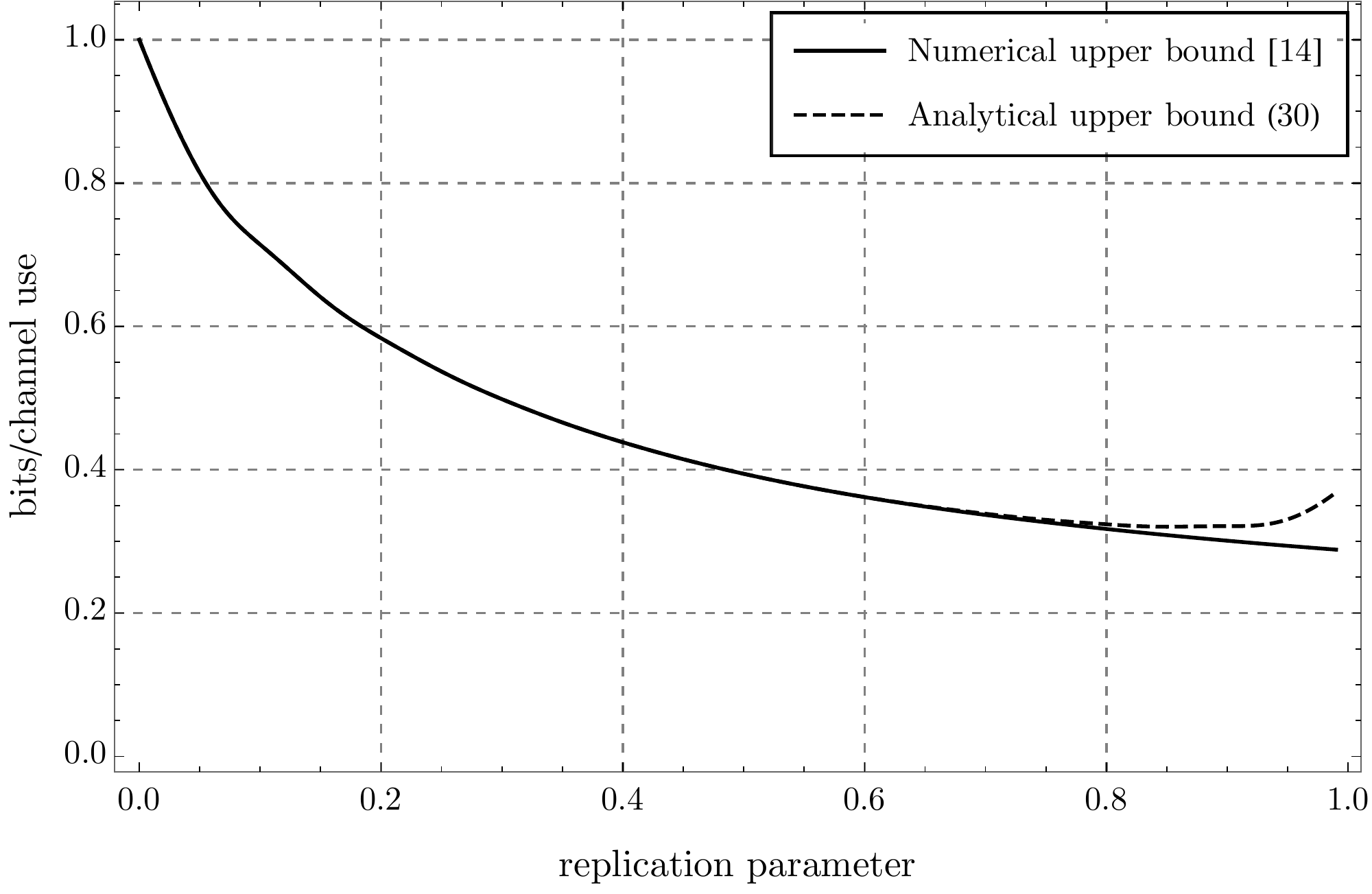}
	\caption{Plot of the numerical capacity upper bound from \cite{MTL12} for the geometric sticky channel and the analytical capacity upper bound \eqref{eq:boundgeomsticky}.}
	\label{fig:geomstickycomp}
\end{figure}

\begin{figure}
	\centering
	\includegraphics[width=0.7\textwidth]{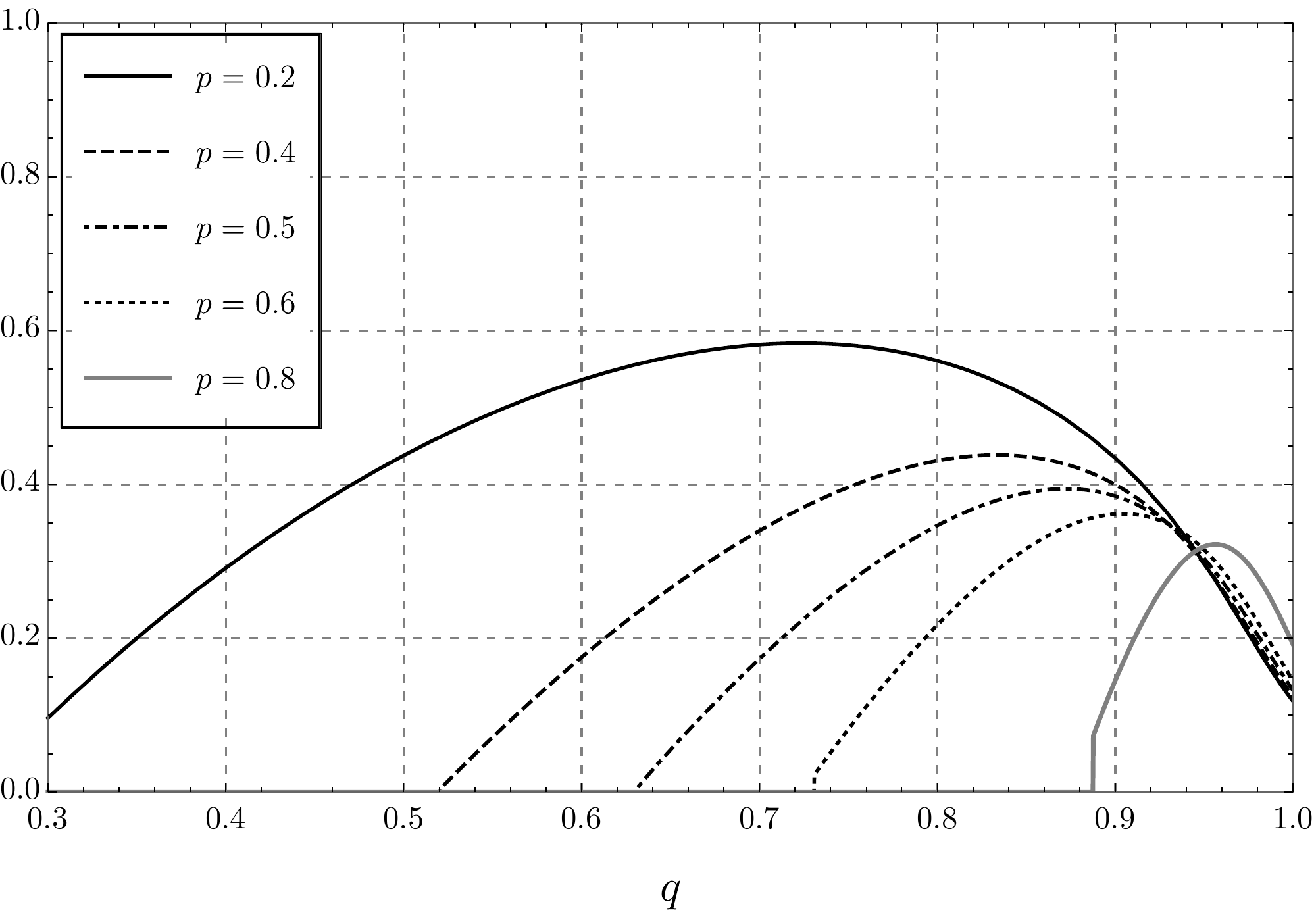}
	\caption{Function inside the supremum in \eqref{eq:boundgeomsticky} for some values of $p$. The zone where the function is zero corresponds to the cases where $\mathds{E}[\Yq]<\frac{1}{1-p}$.}
	\label{fig:geomstickycurves}
\end{figure}

\section{The elementary duplication channel}\label{sec:dupl}

In this section, we study the capacity of the elementary duplication channel. Recall that this channel duplicates each input bit with some probability $p$. More precisely, each input is duplicated according to the distribution $D$ satisfying
\[
D(y)=\begin{cases}
1-p,&\text{ if $y=1$},\\
p,& \text{ if $y=2$},\\
0,&\text{ else.}
\end{cases}
\]

By \eqref{eq:stickychar}, it suffices to study the capacity of the channel which on input $x\in\{1,2,\dots\}$ outputs
\[
Y_x=\sum_{i=1}^x D_i,
\]
where the $D_i$ are i.i.d.\ according to $D$. It follows that
\[
Y_x=x+\mathsf{Bin}_{x,p},
\]
where $\mathsf{Bin}_{x,p}$ denotes the binomial distribution with $x$ trials and success probability $p$; i.e.,
\[
\mathsf{Bin}_{x,p}(y)=\binom{x}{y}(1-p)^{x-y}p^y,\quad y=0,1,\dots,x.
\]
As a result, we have
\begin{equation}\label{eq:yxdupl}
Y_x(y)=\binom{x}{y-x}(1-p)^{2x-y}p^{y-x},\quad y=x,x+1,\dots,2x.
\end{equation}

Our results in this section have a similar flavor to those obtained for the geometric sticky channel in Section \ref{sec:geomsticky}. In particular, we analytically derive a distribution for Theorem~\ref{thm:duality} with zero KL-gap for the elementary duplication channel, and compare the capacity upper bounds obtained with the previously known numerical capacity bounds.

\subsection{A distribution with zero KL-gap for the elementary duplication channel}\label{sec:zerogapdupl}

In this section, we derive a family of distributions with zero KL-gap for the elementary duplication channel. The reasoning behind their design is very similar to what was already discussed in detail in Section \ref{sec:geomsticky}. As a result, we will keep this discussion short.

Recall \eqref{eq:yxdupl} and note that $\mathds{E}[Y_x]=x(1+p)$. Then,
\begin{align*}
-H(Y_x)&=\mathds{E}\left[\log\binom{x}{Y_x-x}\right]-\mathds{E}[Y_x]\frac{h(p)}{1+p}\\
&=\log x! - \mathds{E}[\log(Y_x-x)!]-\mathds{E}[\log(2x-Y_x)!] -\mathds{E}[Y_x]\frac{h(p)}{1+p}.
\end{align*}
By Lemma~\ref{lem:intlog}, and observing that the probability generating function of $Y_x$ is
\[
(z(1-p+pz))^x,
\] 
we have\footnote{Note that we do not require Lemma~\ref{lem:fubini} in this section, since the sum in the expected value ranges only over a finite set.}
\begin{align*}
\log x!&=\int_0^1 \frac{1-tx-(1-t)^x}{t\log(1-t)}dt,\\
\mathds{E}[\log(Y_x-x)!]&=\int_0^1 \frac{1-txp-(1-pt)^x}{t\log(1-t)}dt,\\
\mathds{E}[\log(2x-Y_x)!]&=\int_0^1 \frac{1-tx(1-p)-(1-(1-p)t)^x}{t\log(1-t)}dt.
\end{align*}
Consider the functions
\begin{align}\label{eq:fdupl}
f_1(y,t)&=\frac{1-\frac{ty}{1+p}}{t\log(1-t)}-\frac{\left(\sqrt{(1+p)^2-4pt}-(1-p)\right)^y}{(2p)^{y}t\log(1-t)},\nonumber\\
f_2(y,t)&=\frac{1-\frac{typ}{1+p}}{t\log(1-t)}-\frac{\left(\sqrt{(1+p)^2-4p^2t}-(1-p)\right)^y}{(2p)^{y}t\log(1-t)},\\
f_3(y,t)&=\frac{1-\frac{ty(1-p)}{1+p}}{t\log(1-t)}-\frac{\left(\sqrt{(1+p)^2-4p(1-p)t}-(1-p)\right)^y}{(2p)^{y}t\log(1-t)}.\nonumber
\end{align}
It is straightforward to see, using the probability generating function of $Y_x$, that
\begin{align*}
\mathds{E}[f_1(Y_x,t)]&=\frac{1-tx-(1-t)^x}{t\log(1-t)},\\
\mathds{E}[f_2(Y_x,t)]&=\frac{1-txp-(1-pt)^x}{t\log(1-t)},\\
\mathds{E}[f_3(Y_x,t)]&=\frac{1-tx(1-p)-(1-(1-p)t)^x}{t\log(1-t)}.
\end{align*}
Let
\begin{equation}\label{eq:lambdadupl}
\Lambda_i(y)=\int_0^1 f_i(y,t)dt, \quad i=1,2,3
\end{equation}
with $f_i$ defined as in \eqref{eq:fdupl} for $i=1,2,3$. The functions $f_i(p,y,\cdot)$ are clearly continuous in $(0,1)$ for $i=1,2,3$. Furthermore, they have finite limits when $t\to 0$ and $t\to 1$. This means they can be extended by continuity to $[0,1]$, and so are bounded in $(0,1)$. This is enough to guarantee that
\begin{equation}\label{eq:gdupl}
g(y)=\Lambda_1(y)-\Lambda_2(y)-\Lambda_3(y)
\end{equation}
satisfies
\begin{equation}\label{eq:gpropdupl}
\mathds{E}[g(Y_x)]=\log x!-\mathds{E}[\log(Y_x-x)!]-\mathds{E}[\log(2x-Y_x)!]
\end{equation}
for all $x\geq 1$.

It is possible to see, via an argument very similar to the one used in the proof of Lemma~\ref{lem:asymp}, that $\Yq(y)/y_0=\Theta(q^y/\sqrt{y})$, and so
\[
\Yq(y)=y_0 q^y \exp(g(y)-yh(p)/(1+p))
\]
is a valid distribution (i.e., the normalizing factor $y_0$ exists) exactly when $q\in(0,1)$. 


Finally, recalling \eqref{eq:gpropdupl}, we have
\[
\KL(Y_x||\Yq)=-\log y_0-\mathds{E}[Y_x]\log q,
\]
and so $\Yq$ is a valid distribution which has zero KL-gap for the elementary duplication channel.

As a result, we obtain the capacity upper bound
\begin{equation}\label{eq:boundmudupl}
\Ca'_\mu(D)\leq \inf_{q\in(0,1)}(-\log y_0-\mu\log q)
\end{equation}
via Theorem~\ref{thm:duality} for all $\mu\geq 1$.

\subsection{Capacity upper bound for the elementary duplication channel}\label{sec:bounddupl}

In this section, we derive an analytical capacity upper bound for the elementary duplication channel obtained by combining \eqref{eq:stickychar} with Theorem~\ref{thm:duality} and the family of distributions $\Yq$ from Section \ref{sec:zerogapdupl}, and compare it to the numerical capacity bounds from \cite{Mit08}.

Fix $p\in(0,1)$. We begin by observing that, in this case, we have $\lambda=\mathds{E}[D]=1+p$. The bound follows by combining \eqref{eq:stickychar} with Theorem~\ref{thm:duality} and \eqref{eq:boundmudupl}, summarized below.

\begin{coro}
	For every $p\in (0,1)$, we have
	\begin{align}
	\Ca(D)&\leq(1+p)\sup_{\mu\geq 1+p} \frac{\inf_{q\in(0,1)}(-\log y_0-\mu\log q)}{\mu}\\
	&\leq \sup_{q\in(0,1):\atop\mathds{E}[\Yq]\geq 1+p}\frac{(1+p)(-\log y_0-\mathds{E}[\Yq]\log q)}{\mathds{E}[\Yq]},\label{eq:bounddupl}
	\end{align}
	where
	\begin{align*}
	1/y_0&=\sum_{y=1}^\infty q^y e^{\Lambda_1(y)-\Lambda_2(y)-\Lambda_3(y)-yh(p)/(1+p)}<\infty,\\
	\mathds{E}[\Yq]&=\sum_{y=1}^\infty y_0yq^y e^{\Lambda_1(y)-\Lambda_2(y)-\Lambda_3(y)-yh(p)/(1+p)},
	\end{align*}
	with $\Lambda_1$, $\Lambda_2$, and $\Lambda_3$ defined as in \eqref{eq:lambdadupl}.
\end{coro}


As in \eqref{eq:boundgeomsticky}, we obtain \eqref{eq:bounddupl} by choosing, for each $\mu\geq 1+p\geq 1$, the value of $q$ such that $\mathds{E}[\Yq]=\mu$. This is guaranteed by the fact that, similarly to Lemma~\ref{lem:asymp}, we have $\Yq(y)/y_0=\Theta(q^y/\sqrt{y})$, as was already mentioned.

Table \ref{table:compdupl} compares the analytical capacity upper bound obtained via \eqref{eq:bounddupl} with the explicit data points of the numerical bounds for the elementary duplication channel in \cite{Mit08}, which are rounded to four decimal digits. Figure~\ref{fig:duplcomp} plots the numerical capacity upper bound from \cite{Mit08} and the analytical upper bound \eqref{eq:bounddupl}.

Unlike the capacity upper bound we obtained for the geometric sticky channel (see Section \ref{sec:boundsticky}), we see that \eqref{eq:bounddupl} is only tight for small $p$, and becomes trivial if $p$ is too large. Nevertheless, we are still able to improve on the numerical upper bound from \cite{Mit08} for, say, $p=0.2$.

\begin{table}[h]
	\caption{Comparison between the numerical capacity bounds for the elementary duplication channel from \cite{Mit08} and the upper bound \eqref{eq:bounddupl} in bits/channel use.}\label{table:compdupl}
	\centering
	\begin{tabular}{ |c|c|c|c| }
		\hline
		$p$ & Lower bound \cite{Mit08} & Upper bound \cite{Mit08} & Upper bound  \eqref{eq:bounddupl}\\ 
		\hline
		0.1 & 0.7405& 0.7406 &0.7406\\ 
		\hline
		0.2 & 0.6611 & 0.6618 &0.6611\\ 
		\hline
		0.3 & 0.6400 &  0.6404 & 0.6419\\ 
		\hline
		0.4 & 0.6488 & 0.6499 &0.6625\\ 
		\hline
		0.5 & 0.6788 & 0.6797 &0.7182\\ 
		\hline
		0.6 & 0.7273 & 0.7277 &0.8126\\ 
		\hline
		0.7 & 0.7914 & 0.7915 &0.9553\\ 
		\hline
		0.8 & 0.8674 & 0.8675 &$>1$\\ 
		\hline
		0.9 & 0.9469 & 0.9479 &$>1$\\ 
		\hline
	\end{tabular}
\end{table}

\begin{figure}
	\centering
	\includegraphics[width=0.7\textwidth]{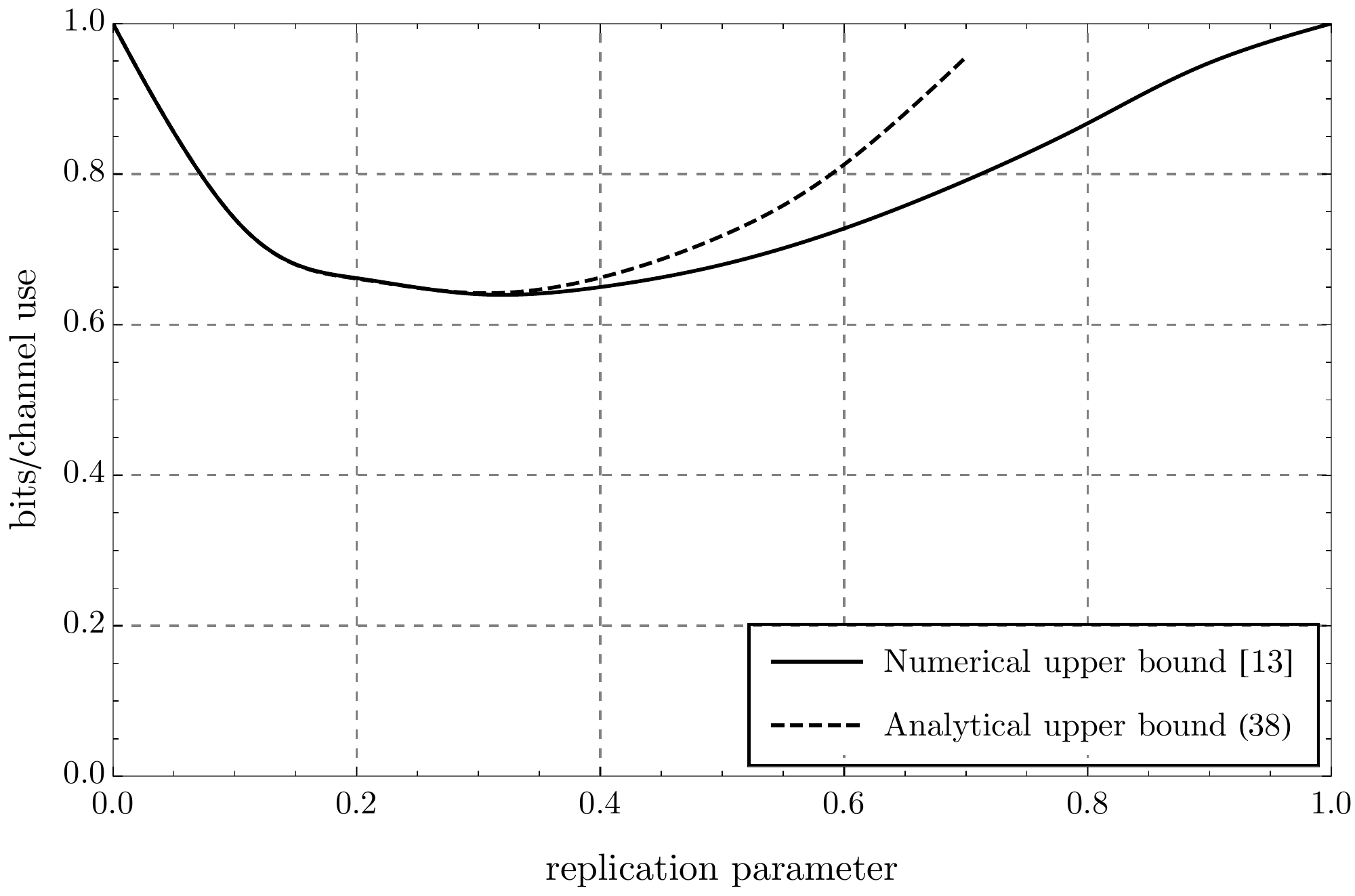}
	\caption{Plot of the numerical capacity upper bound from \cite{Mit08} for the elementary duplication channel and the analytical capacity upper bound \eqref{eq:bounddupl}.}
	\label{fig:duplcomp}
\end{figure}

\begin{figure}
	\centering
	\includegraphics[width=0.7\textwidth]{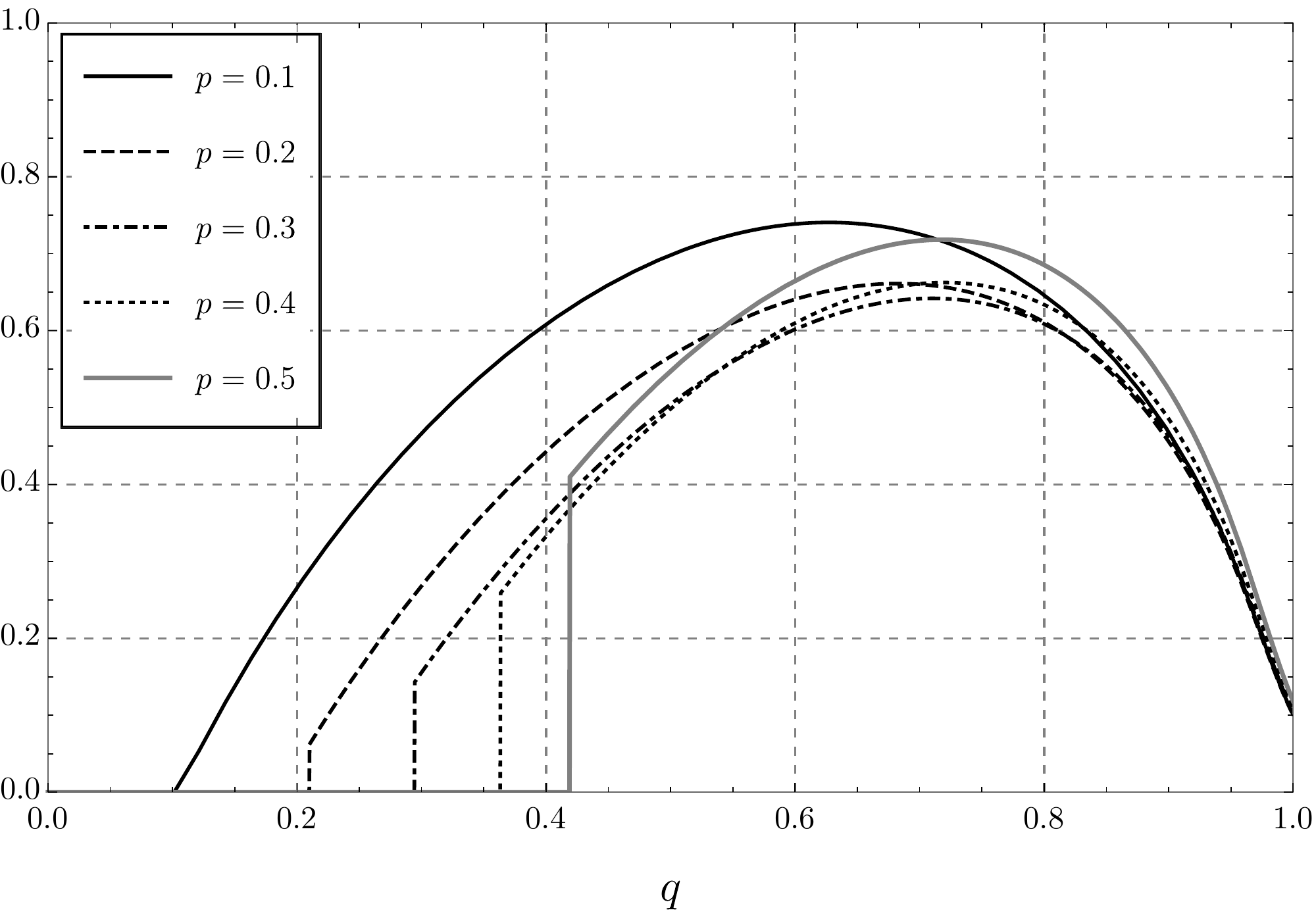}
	\caption{Function inside the supremum in \eqref{eq:bounddupl} for some values of $p$. The zone where the function is zero corresponds to the cases where $\mathds{E}[\Yq]<1+p$.}
	\label{fig:duplcurves}
\end{figure}

\section{Geometric replications with deletions}\label{sec:geomdel}

In this section, we study the capacity of a channel that combines deletions with geometric replications, which we call the \emph{geometric deletion channel}. This channel independently replicates each input bit according to a geometric distribution with support on $\{0,1,2,\dots\}$. More precisely, each input bit $x_i$ is replaced by $D_{0i}$ copies of its value at the output, where the $D_{0i}$ are i.i.d.\ according to $D_0$ satisfying
\[
D_0(y)=(1-p)p^y,\quad y=0,1,2,\dots
\]
where $p$ is the replication parameter. Recall that in the model from \cite{MTL12}, this channel corresponds to the case where $p_d=1-p_t$, where $p_d$ is the deletion probability and $p_t$ is the replication probability in any given round.

We specialize Theorem~\ref{thm:red} for the $0$-geometric channel. In this case we have $\overline{D_0}=1+D_0$, and as a result, for $x\in\{0,1,2,\dots\}$,
\[
\overline{D_0}+\sum_{i=1}^x D_{0i} = 1+\NB_{1+x,p},
\]
where, as before, $\NB_{r,p}$ denotes the negative binomial distribution with $r$ failures and success probability $p$.

Therefore, $\Ch'_\mu(D_0)$ is the mean-limited channel which on input $x\in\{1,2,\dots\}$ outputs
\[
Z_x=1+\NB_{x,p}.
\]
For convenience, we will work with a slightly modified channel. Note that the capacity of $\Ch'_\mu(D_0)$ is equal to the capacity of the channel $\Ch''_{\mu-1}(D_0)$ which on input $x\in\{1,2,\dots\}$ outputs 
\[
Y_x=\NB_{x,p}
\]
with output mean constraint $\mu-1$. We name this channel the \emph{negative binomial channel}. The output mean constraint changes from $\mu$ to $\mu-1$ because for the same input $x$ we have $Y_x=Z_x-1$. Finally, note that $\lambda=\mathds{E}[D]=\frac{p}{1-p}$, $\overline{\lambda}=\mathds{E}[\overline{D}]=\frac{1}{1-p}$. Letting $\Ca''_{\mu-1}(D_0)$ denote the capacity of $\Ch''_{\mu-1}(D_0)$ yields the following specialized version of Theorem~\ref{thm:duality}.
\begin{coro}\label{cor:geombound}
	We have
	\begin{equation*}
	\Ca(D_0)\leq \frac{p}{1-p}\sup_{\mu\geq 1/(1-p)}\frac{\Ca''_{\mu-1}(D_0)}{\mu}.
	\end{equation*}
\end{coro}

In the following sections we will focus on upper bounding the capacity of the negative binomial channel via Theorem~\ref{thm:duality}.

\subsection{A bound via convexity}\label{sec:firstbound}

In this section, we obtain a capacity upper bound for the negative binomial channel by following a reasoning similar to the one used to derive capacity upper bounds for the deletion channel in \cite{Che17}. For convenience, we define $d=1-p$.

As previously observed, we can write
\begin{align}\label{eq:KL}
D_\mathsf{KL}(Y_x||Y)&=\sum_{y=0}^\infty Y_x(y)\log\frac{Y_x(y)}{Y(y)}\nonumber\\
&=-H(Y_x)-\sum_{y=0}^\infty Y_x(y)\log Y(y).
\end{align}
Furthermore, recalling that $Y_x= \NB_{x,p}$ and from the fact that $\mathds{E}[Y_x]=\frac{xp}{1-p}$, we have
\begin{align}\label{eq:entNB}
-H(Y_x)&=\sum_{y=0}^\infty Y_x(y)\log\left(\binom{y+x-1}{y}d^x p^y\right) \nonumber \\
&=\mathds{E}\left[\log\binom{Y_x+x-1}{Y_x}\right]+x\log d +\mathds{E}[Y_x]\log p\nonumber\\
&=\mathds{E}\left[\log\binom{Y_x+x-1}{Y_x}\right]-\mathds{E}[Y_x]\frac{h(p)}{p}.
\end{align}

We consider a family of distributions $Y^{(q)}$ for $q\in(0,1)$ of the form
\begin{equation*}
Y^{(q)}(y)=y_0\binom{g(y)}{y}q^y \exp(-yh(p)/p),\quad y=0,1,2,\dots
\end{equation*}
for a function $g$ to be defined, where 
\begin{equation}\label{eq:normfactor}
y_0=\left(\sum_{y=0}^\infty Y^{(q)}(y)/y_0\right)^{-1}
\end{equation}
is the normalizing factor. Instantiating $Y$ with $Y^{(q)}$ leads to
\begin{align}\label{eq:KLY}
D_\mathsf{KL}(Y_x||Y^{(q)})&=\mathds{E}\left[\log\binom{Y_x+x-1}{Y_x}\right]+\mathds{E}[Y_x]\frac{h(p)}{p} \nonumber\\ 
&-\sum_{y=0}^\infty Y_x(y)\log\left(y_0\binom{g(y)}{y}q^y \exp(-yh(p)/p)\right)\nonumber\\
&=\mathds{E}\left[\log\frac{\binom{Y_x+x-1}{Y_x}}{\binom{g(Y_x)}{Y_x}}\right]-\log y_0 -\mathds{E}[Y_x]\log q.
\end{align}

Equipped with some insight, we want to choose $g$ such that
\begin{equation}\label{eq:gcond}
g(\mathds{E}[Y_x])=\mathds{E}[Y_x]+x-1,
\end{equation}
which can be accomplished by setting $g(y)=y/p-1$. This leads to the expression
\begin{equation}\label{eq:Yq}
Y^{(q)}(y)=y_0\binom{y/p-1}{y}q^y \exp(-yh(p)/p)
\end{equation}
It is straightforward to see that $\Yq$ is a valid distribution for all $q\in(0,1)$, i.e., $1/y_0<\infty$, by using the asymptotic expression for $\binom{y/p-1}{y}$ obtained via Stirling's approximation.

Combining \eqref{eq:KLY} and \eqref{eq:Yq}, we obtain
\begin{equation}\label{ineq:KLline}
D_\mathsf{KL}(Y_x||Y^{(q)})\leq -\eps(p)-\log y_0-\mathds{E}[Y_x]\log q
\end{equation}
for all integers $x\geq 1$, where
\begin{equation}\label{eq:KLgap}
\eps(p)=\inf_{x\geq 1} \mathds{E}\left[\log\frac{\binom{Y_x/p-1}{Y_x}}{\binom{Y_x+x-1}{Y_x}}\right].
\end{equation}

As we shall see, we can always replace $\eps(p)$ by $0$ in \eqref{ineq:KLline} to obtain a valid upper bound. In order to prove this, we first need an auxiliary result from \cite{KP16}.
\begin{lem}[\protect{\cite[Lemma~1, specialized]{KP16}}]\label{lem:loggammaquotient}
	Consider the function
	\[
	f(y)=\log\left(\frac{\prod_{i=1}^{k_1}\Gamma(A_i y+a_i)}{\prod_{j=1}^{k_2}\Gamma(B_j y+b_j)}\right).
	\]
	Then, $f$ is convex in $(0,\infty)$ provided that
	\[
	\sum_{i=1}^{k_1}\frac{\exp(-a_i u/A_i)}{1-\exp(-u/A_i)}-\sum_{i=1}^{k_2}\frac{\exp(-b_j u/B_j)}{1-\exp(-u/B_j)}\geq 0
	\]
	for all $u> 0$.
\end{lem}

We are now ready to prove the desired result.
\begin{lem}\label{lem:epsneg}
	We have $\eps(p)\geq 0$ for all $p\in(0,1)$.
\end{lem}
\begin{proof}
	We show that $f_x(y)=\log\left[ \binom{y/p-1}{y}/\binom{y+x-1}{y}\right]$ is convex in $[0,\infty)$ for all $x\geq 1$. This implies the desired result via Jensen's inequality, since, by the choice of $g$ (recall \eqref{eq:gcond}), we have
	\[
		\mathds{E}[f_x(Y_x)]\geq f_x(\mathds{E}[Y_x])=\log 1=0.
	\]
	Therefore, $\eps(p)=\inf_{x\geq 1}\mathds{E}[f_x(Y_x)]\geq 0$.	

	By Lemma~\ref{lem:loggammaquotient}, showing that $f_x$ is convex in $(0,\infty)$ boils down to showing that
	\[
	P_x(u)=\frac{1}{1-e^{-up}}-\frac{1}{1-e^{-up/(1-p)}}-\frac{e^{-ux}}{1-e^{-u}}\geq 0
	\]
	for all $x\geq 1$ and $u>0$. Note that $P_x(u)\geq P_1(u)$ for $x\geq 1$, and that $P_1(u)$ can be rewritten as
	\[
	P_1(u)=\frac{1}{e^{up}-1}-\frac{1}{e^{up/(1-p)}-1}-\frac{1}{e^{u}-1}.
	\]
	Therefore, it suffices to show that
	\begin{equation}\label{eq:split1}
	\frac{1-p}{e^{up}-1}-\frac{1}{e^{up/(1-p)}-1}\geq 0
	\end{equation}
	and
	\begin{equation}\label{eq:split2}
	\frac{p}{e^{up}-1}-\frac{1}{e^{u}-1}\geq 0.
	\end{equation}
	We show only \eqref{eq:split2}, and observe that \eqref{eq:split1} follows in an analogous manner. Rearranging, we want to show that
	\begin{equation}\label{eq:rearrange}
	p(e^u-1)-(e^{up}-1)\geq 0.
	\end{equation}
	Note that the left-hand side of \eqref{eq:rearrange} is 0 at $u=0$, and that its derivative with respect to $u$ is
	\[
	p(e^u-e^{up}),	
	\]
	which is positive for all $u>0$. This yields the desired inequality.

It remains to see that $f_x$ is convex in $[0,\infty)$. Note that $f_x(0)=0$, since $\binom{-1}{0}=1$. Furthermore,
\[
	\lim_{y\to 0^+} \log\binom{y/p-1}{y}=\lim_{y\to 0^+} \log\left(d\binom{y/p}{y}\right)=\log d<0.
\]
This implies that $\lim_{y\to 0^+}f_x(y)=\log d<0$ for all $x\geq 1$. We then have $f_x(0)=0>\lim_{y\to 0}f_x(y)$, which shows that $f_x$ is convex in $[0,\infty)$ (recall we had already shown it was convex in $(0,\infty)$).
\end{proof}

While Lemma~\ref{lem:epsneg} implies that we can replace $\eps(p)$ by 0 in \eqref{ineq:KLline}, it turns out that $\eps(p)$ is actually significantly larger than zero for most values of $p$, and so keeping it in \eqref{ineq:KLline} leads to improved capacity upper bounds for the negative binomial channel.

We are now in a position to apply Theorem~\ref{thm:duality} using \eqref{ineq:KLline}.
\begin{thm}\label{thm:firstUB}
	We have
	\begin{align}
	\Ca''_{\mu}(D_0)&\leq -\eps(p)+\inf_{q\in(0,1)}(-\log y_0-\mu\log q)\label{eq:firstboundeps}\\
	&\leq \inf_{q\in(0,1)}(-\log y_0-\mu\log q).
	\end{align}
\end{thm}

Interestingly, $Y^{(q)}$ is very closely related to the \emph{inverse binomial distribution} defined in \cite{Che17} to obtain capacity upper bounds for the deletion channel. For given $p,q\in(0,1)$, we denote the inverse binomial distribution by $\mathsf{InvBin}_{p,q}$. It satisfies
\begin{equation}
\mathsf{InvBin}_{p,q}(y)=y_{\mathsf{IB}}\binom{y/p}{y}q^y\exp(-y h(p)/p),
\end{equation}
where $y_{\mathsf{IB}}$ is the normalizing factor. Using the equality
\[
\binom{y/p-1}{y}=d\binom{y/p}{y}
\]
valid for all $y\geq1$ and recalling \eqref{eq:Yq}, we conclude that
\begin{equation}\label{eq:relib}
\frac{Y^{(q)}}{y_0}=d\cdot\frac{\mathsf{InvBin}_{p,q}(y)}{y_{\mathsf{IB}}}
\end{equation}
for all $y\geq 1$. This property of $Y^{(q)}$ will prove to be very useful in the following sections, as the normalizing factor and expected value of $\mathsf{InvBin}_{p,q}$ are well understood in terms of both special and elementary functions.

\subsection{A bound via truncation}\label{sec:trunc}

In this section, we design a distribution whose KL-gap converges to 0 exponentially fast as $x$ increases. The process will be similar to that of Sections~\ref{sec:zerogapsticky}~and~\ref{sec:zerogapdupl}, and we will reutilize some arguments. As was the case for the deletion and Poisson-repeat channels in~\cite[Sections 5 and 6]{Che17}, in this case we cannot ensure that the KL-gap is zero. 


We consider a family of distributions $\YYq$, for $q\in(0,1)$, of the form
\begin{equation}\label{eq:yq0geom}
\YYq(y)=\overline{y_0} q^y\exp(g(y)-yh(p)/p), \quad y=0,1,2,\dots
\end{equation}
for some function $g$ to be determined, where $y_0$ is the normalizing factor. Recalling that $Y_x=\NB_{x,p}$ and \eqref{eq:entNB}, we want $g$ to satisfy
\begin{align}
\mathds{E}[g(Y_x)]&=\mathds{E}[\log(Y_x+x-1)!]-\log(x-1)!-\mathds{E}[\log Y_x!]+ R_p(x),
\end{align}
where $R_p(x)\geq 0$ is an error term which vanishes exponentially fast with $x$. Furthermore, we want $g$ to have moderate growth so that $\YYq$ is a valid probability distribution. We note that $g(y)$ can grow at most like $y\log y +O(y)$.

Recalling Lemma~\ref{lem:intlog}, we have
\[
\log(x-1)!=\int_0^1\frac{1+t-tx-(1-t)^{x-1}}{t\log(1-t)}dt
\]
and\footnote{Once again, switching the integral and expected value in \eqref{eq:explog0geom} is allowed via Lemma~\ref{lem:fubini}, since the function inside the integral is continuous in $[0,1]$ and positive for all $y\geq 0$ and $x\geq 1$.}
\begin{align}\label{eq:explog0geom}
&\mathds{E}[\log(Y_x+x-1)!]=\mathds{E}\left[\int_0^1\frac{1+t-t(y+x)-(1-t)^{y+x-1}}{t\log(1-t)}dt\right]\nonumber\\
&=\int_0^1\frac{1+t-\frac{tx}{1-p}-(1-t)^{x-1}\left(\frac{1-p}{1-p(1-t)}\right)^x}{t\log(1-t)}dt.
\end{align}

Consider the functions
\begin{align*}
f_1(y,t)&=\frac{1+t-ty(1-p)/p-\left(\frac{p-t}{p(1-t)}\right)^y/(1-t)}{t \log(1-t)},\\
f_2(y,t)&=\frac{1+t-ty/p-\left(\frac{p-t(1+p)}{p(1-t)}\right)^y/(1-t)}{t \log(1-t)}.
\end{align*}
It holds that
\begin{align*}
\mathds{E}[f_1(Y_x,t)]&=\frac{1+t-tx-(1-t)^{x-1}}{t\log(1-t)},\\
\mathds{E}[f_2(Y_x,t)]&=\frac{1+t-\frac{tx}{1-p}-(1-t)^{x-1}\left(\frac{1-p}{1-p(1-t)}\right)^x}{t\log(1-t)}.
\end{align*}
As a result, we would hope that
\begin{align*}
\mathds{E}\left[\int_0^1 f_1(Y_x,t)dt\right]&=\log(x-1)!,\quad \forall x\geq 1\\
\mathds{E}\left[\int_0^1 f_2(Y_x,t)dt\right]&=\mathds{E}[\log(Y_x+x-1)!],\quad \forall x\geq 1.
\end{align*}
However, this does not hold as the above integrals on the left-hand side diverge. This means that the unique formal solutions to the functional equations above are not well-defined functions, as was the case for the analogous equation associated to the Poisson-repeat channel in~\cite{Che17}. The formal solutions for the analogous functional equations in the case of the binary deletion channel in \cite{Che17} are well-defined, but do not lead to a valid distribution. We can contrast this with the geometric sticky and elementary duplication channels in Sections~\ref{sec:geomsticky}~and~\ref{sec:dupl}, where we derive such analogous formal solutions and prove that they are well-defined and lead to a valid distribution.


In order to overcome this, we truncate the integration bounds. To determine the point at which to truncate, note that $\frac{p-t(1+p)}{p(1-t)}\geq -1$ whenever $t\leq \frac{2p}{1+2p}$. Truncating at this point ensures that the exponential terms in the two integrals are controlled. Consider the truncated integrals
\begin{align}
\Lambda_1(y)&=\int_0^{\frac{2p}{1+2p}} f_1(y,t)dt\label{eq:lambda10geom}\\
\Lambda_2(y)&=\int_0^{\frac{2p}{1+2p}} f_2(y,t)dt.\label{eq:lambda20geom}
\end{align}
An argument similar to that used in the proof of Lemma~\ref{lem:switch} shows that both $f_1(y,\cdot)$ and $f_2(y,\cdot)$ are non-negative in $\left(0,\frac{2p}{1+2p}\right]$ for large enough $y$. It is also easy to see that $f_1(y,\cdot)$ and $f_2(y,\cdot)$ are continuous in $\left(0,\frac{2p}{1+2p}\right]$, and that they can be extended by continuity to $\left[0,\frac{2p}{1+2p}\right]$. This means that the conditions of Lemma~\ref{lem:fubini} are satisfied, and so
\begin{align*}
\mathds{E}[\Lambda_1(Y_x)]&=\int_0^{\frac{2p}{1+2p}} \mathds{E}[f_1(y,t)]\\
&=\log(x-1)!-\int_{\frac{2p}{1+2p}}^1 \frac{1+t-tx-(1-t)^{x-1}}{t\log(1-t)}\\
&=\log(x-1)!-\eta\left(\frac{1}{1+2p}\right)+(x-1)\text{Li}\left(\frac{1}{1+2p}\right)+\int_{\frac{2p}{1+2p}}^1 \frac{(1-t)^{x-1}}{t\log(1-t)},
\end{align*}
where $\text{Li}(z)=\int_0^z \frac{dt}{\log t}$ is the logarithmic integral and $\eta(z)=\int_0^z \frac{dt}{(1-t)\log t}$. Analogously,
\begin{align*}
\mathds{E}[\Lambda_2(Y_x)]&=\int_0^{\frac{2p}{1+2p}} \mathds{E}[f_2(y,t)]\\
&=\mathds{E}[\log(Y_x+x-1)!]-\int_{\frac{2p}{1+2p}}^1 \frac{1+t-\frac{tx}{1-p}-(1-t)^{x-1}\left(\frac{1-p}{1-p(1-t)}\right)^x}{t\log(1-t)}\\
&=\mathds{E}[\log(Y_x+x-1)!]-\eta\left(\frac{1}{1+2p}\right)+\left(\frac{x}{1-p}-1\right)\text{Li}\left(\frac{1}{1+2p}\right)+\int_{\frac{2p}{1+2p}}^1 \frac{(1-t)^{x-1}\left(\frac{1-p}{1-p(1-t)}\right)^x}{t\log(1-t)}.
\end{align*}

We set
\[
g(y)=\Lambda_2(y)-\Lambda_1(y)-\log y!-y\text{Li}\left(\frac{1}{1+2p}\right).
\]
Note that $g$ satisfies
\begin{align}\label{eq:g0geom}
\mathds{E}[g(Y_x)]&=\mathds{E}[\log(Y_x+x-1)!]-\log(x-1)!-\mathds{E}[\log Y_x!]-\mathds{E}[Y_x]\text{Li}\left(\frac{1}{1+2p}\right)+R_p(x),
\end{align}
where
\begin{equation}\label{eq:rp}
R_p(x)=-\int_{\frac{2p}{1+2p}}^1 \frac{(1-t)^{x-1}\left(1-\left(\frac{1-p}{1-p(1-t)}\right)^x\right)}{t\log(1-t)}\geq 0.
\end{equation}
Observe that $R_p(x)$ vanishes exponentially fast in $x$.

It now remains to show that $g$ has the correct asymptotic growth. The proof of the following result is analogous to the proof of Lemma~\ref{lem:asymp}.
\begin{lem}\label{lem:asymp0geom}
	We have
	\begin{align*}
	\Lambda_1(y)&=\log\Gamma\left(\frac{y(1-p)}{p}\right)+\frac{y(1-p)}{p}\cdot\textnormal{Li}\left(\frac{1}{1+2p}\right) -\eta\left(\frac{1}{1+2p}\right)+O(1),\\
	\Lambda_2(y)&=\log\Gamma\left(\frac{y}{p}\right)+\frac{y}{p}\cdot\textnormal{Li}\left(\frac{1}{1+2p}\right)-\eta\left(\frac{1}{1+2p}\right)+O(1).
	\end{align*}
	In particular,
	\[
	y\frac{h(p)}{p}-g(y)=\frac{1}{2}\log y +O(1).
	\]
\end{lem}

Lemma~\ref{lem:asymp0geom} implies that $\YYq$ is a valid distribution if $q\in(0,1)$, since it shows that $\YYq/\overline{y_0}=\Theta(q^y/\sqrt{y})$. It remains to upper bound $\KL(Y_x||\YYq)$. We have
\begin{align}\label{eq:line0geom}
\KL(Y_x||\YYq)&=-H(Y_x)-\sum_{y=0}^\infty Y_x(y)\log \YYq(y)\nonumber\\
&=-H(Y_x)-\log \overline{y_0}- \mathds{E}[Y_x]\log q - \mathds{E}[g(Y_x)]+\mathds{E}[Y_x]h(p)/p\nonumber\\
&=\mathds{E}\left[\log\binom{Y_x+x-1}{Y_x}\right]-\log \overline{y_0}-\mathds{E}[Y_x]\log q-\mathds{E}[g(Y_x)]\nonumber\\
&=-R_p(x)-\log\overline{y_0}-\mathds{E}[Y_x]\log q\nonumber\\
&\leq -\log \overline{y_0}-\mathds{E}[Y_x]\log q.
\end{align}
In the above, the second equality follows from \eqref{eq:yq0geom}, the third equality follows from \eqref{eq:entNB}, the fourth equality holds because of \eqref{eq:g0geom}, and the inequality follows from the fact that $R_p(x)\geq 0$ for all $x\geq 1$.

Combining \eqref{eq:line0geom} with Theorem~\ref{thm:duality}, we immediately obtain the capacity upper bound
\begin{equation}\label{eq:firstboundtrunc}
\Ca''_\mu(D_0)\leq \inf_{q\in(0,1)}(-\log \overline{y_0}-\mu\log q)
\end{equation}
for all $\mu\geq 0$. There are two important comments regarding this bound: First, as shown in \eqref{eq:line0geom}, the gap between $\KL(Y_x||\YYq)$ and the line $-\log \overline{y_0}-\mathds{E}[Y_x]\log q$ is exactly $R_p(x)$, which converges to 0 exponentially fast as $x$ increases. Second, we still have $R_p(1)\gg 0$.

\subsection{Improving the bound by fixing the mass at $y=0$}\label{sec:fixmass0}

In this section, we showcase a simple technique which can be used to significantly improve the bounds we obtain from the distributions designed in Sections~\ref{sec:firstbound}~and~\ref{sec:trunc}. We will also use this technique to give a simple proof of an elementary capacity upper bound for the geometric deletion channel with large replication parameter in Section \ref{sec:boundsmalld}. Namely, the capacity in this regime is at most $0.73$ bits/channel use for large replication parameter. As discussed in Section \ref{sec:intro}, this is the first nontrivial elementary capacity upper bound that holds over an interval of the channel parameter for channels with geometric replications and deletions.

The technique we are about to present consists simply in optimizing the mass at $y=0$ of any given family of distributions suitable for Theorem~\ref{thm:duality}. This leads to an upper bound which is at least as good as the original, and, when applied to the distributions from Section \ref{sec:firstbound}, we see significant improvements for a large range of the replication parameter $p$. 


Consider a distribution $Y$ with support on $\{0,1,2,\dots\}$ and probability mass function $Y(y)=y_0 a(y)$ for some function $a(y)$ with $a(0)=1$ and normalizing factor $y_0$. For $\delta\in (0,1]$, consider the modified distribution $Y_\delta$ given by
\begin{equation}
Y_\delta(y)=
\begin{cases}
\alpha\delta,&\text{ if } y=0\\
\alpha a(y),&\text{ if } y>0,
\end{cases}
\end{equation}
where $\alpha$ is the normalizing factor, satisfying $1/\alpha=\delta+1/y_0-1$. Intuitively, $Y_\delta$ is obtained from $Y$ by modifying the mass of $Y$ at $y=0$. Note that setting $\delta=1$ yields the original distribution $Y$.

A key point is that $\KL(Y_x||Y_\delta)$ has a simple expression in terms of $\KL(Y_x||Y)$ for all $x\geq 1$. In fact, letting $d=1-p$ and recalling that $Y_x=\NB_{x,p}$,
 \begin{align}
	\KL(Y_x||Y_\delta)&=-H(Y_x)-\log \alpha-\sum_{y=1}^\infty Y_x(y)\log a(y)-d^x\log\delta\nonumber\\
	&=-H(Y_x)-\log \alpha-\sum_{y=0}^\infty Y_x(y)\log a(y)-d^x\log\delta\nonumber\\
	&=-H(Y_x)-\log y_0-\sum_{y=0}^\infty Y_x(y)\log a(y)+\log y_0-\log \alpha-d^x\log\delta\nonumber\\
	&=\KL(Y_x||Y)+\log y_0-\log \alpha-d^x\log\delta\label{eq:KLdeltaintermediate}\\
	&\leq \KL(Y_x||Y)+\log y_0-\log \alpha-d\log\delta.\label{eq:KLdelta}
\end{align}
In the first equality we used the fact that $Y_x(0)=d^x$ for all $x\geq 1$. The second equality follows because $\log a(0)=0$ since $a(0)=1$. In the last equality we used that $\delta\leq 1$, and so $-d^x\log\delta \leq -d\log\delta$ for $x\geq 1$.

Suppose $\Delta(x)=a\mathds{E}[Y_x]+b-\KL(Y_x||Y)$ is the KL-gap for some fixed line $a\mathds{E}[Y_x]+b$. Then, the new KL-gap between $\KL(Y_x||Y_\delta)$ and the line $a\mathds{E}[Y_x]+b+\log y_0-\log\alpha-d\log\delta$ is
\begin{equation}\label{eq:deltadelta}
	\Delta_\delta(x)=a\mathds{E}[Y_x]+b+\log y_0-\log\alpha-d\log\delta-\KL(Y_x||Y_\delta)=\Delta(x)-d\log\delta+d^x\log\delta\geq 0,
\end{equation}
where the second equality follows from \eqref{eq:KLdeltaintermediate} and the definition of $\Delta(x)$. In particular $\Delta_\delta(1)=\Delta(1)$ and $\Delta_\delta(x)\geq \Delta(x)$. As a result, we have the bound
\begin{align}
	\KL(Y_x||Y_\delta)&\leq \KL(Y_x||Y)+\log y_0-\log \alpha-d\log\delta\nonumber\\
	&\leq a\mathds{E}[Y_x]+b+\log y_0-\log \alpha-d\log\delta-\eps_\delta(p),\label{eq:KLepsdelta}
\end{align}
where
\begin{equation}\label{eq:epsdelta}
	\eps_\delta(p)=\inf_{x\geq 1}\Delta_\delta(x)\geq \eps(p),
\end{equation}
with associated KL-gap
\begin{equation}\label{eq:deltadeltaprime}
	\Delta'_\delta(x)=a\mathds{E}[Y_x]+b+\log y_0-\log \alpha-d\log\delta-\eps_\delta(p)-\KL(Y_x||Y_\delta)=\Delta_\delta(x)-\eps_\delta(p).
\end{equation}

Combined with Theorem~\ref{thm:duality}, this immediately leads to the capacity upper bound
\begin{equation}\label{eq:capUBdelta}
	\Ca''_\mu(D_0)\leq \inf_{q\in(0,1),\delta\in(0,1]}(a\mu+b+\log y_0-\log \alpha-d\log\delta-\eps_\delta(p)).
\end{equation}

Optimizing the right hand side of \eqref{eq:capUBdelta} over two parameters $q$ and $\delta$ is cumbersome. We now argue that a specific choice of $\delta$ works well over a large range of $p$ for the distributions we designed, thus obtaining a much simpler bound than \eqref{eq:capUBdelta} which still gives very good results. As discussed before, as a rule of thumb, a smaller KL-gap leads to improved upper bounds. The distributions we designed in Sections~\ref{sec:firstbound}~and~\ref{sec:trunc} have associated KL-gaps which converge to $0$ when $x\to \infty$ for a large range of $p$. In the case of the truncation-based distribution from Section~\ref{sec:trunc}, this holds for all $p$, and the speed of convergence is exponential in $x$. However, the KL-gap at small $x$ does not behave as well. In general, it is significantly bounded away from $0$ when $x=1$. From experience, the KL-gap at small $x$ appears to have significant influence on the sharpness of the upper bounds obtained.  As a result, it is natural to wonder how one can obtain a small gap for small $x$ without affecting the behavior of the gap for large $x$.

Suppose $\Delta(x)\to L$ when $x\to\infty$, and $\Delta(1)\gg L$. We now describe how we can exploit the method introduced in this section to derive a new upper bound on $\KL(Y_x||Y_\delta)$ with a KL-gap that is $0$ at $x=1$ and converges to $0$ when $x\to\infty$ with a similar speed of convergence to the original KL-gap $\Delta$. Consider $\delta=\exp(-(\Delta(1)-L)/d)$. Then, $\Delta_\delta(1)=\Delta(1)$ and $\Delta_\delta(x)\to \Delta(1)$ when $x\to\infty$. If $\Delta_\delta(x)\geq \Delta(1)$ (which, as we shall see, happens often), we have $\eps_{\delta}(p)=\Delta(1)$, and so, recalling \eqref{eq:KLepsdelta},
\[
	\KL(Y_x||Y_\delta)\leq \KL(Y_x||Y)+\log y_0-\log \alpha-L\leq a\mathds{E}[Y_x]+b+\log y_0-\log \alpha-L
\]
with corresponding KL-gap (recall~\eqref{eq:deltadeltaprime})
\[
	\Delta'_\delta(x)=a\mathds{E}[Y_x]+b+\log y_0-\log \alpha-L-\KL(Y_x||Y_\delta)=\Delta_\delta(x)-\Delta(1)
\]
satisfying $\Delta'_\delta(1)=0$ and $\Delta'_\delta(x)\to 0$ when $x\to\infty$ with an exponentially small penalty in the speed of convergence, as desired.

We instantiate the reasoning just described with the distributions designed in Sections~\ref{sec:firstbound}~and~\ref{sec:trunc}. Consider $\YYq$ from Section~\ref{sec:trunc}. We will use overlines over the relevant quantities associated to $\YYq$ to distinguish from the same quantities associated to $\Yq$ from Section~\ref{sec:firstbound}. Recalling~\eqref{eq:line0geom}, let
\[
	\overline\Delta(x)=-\mathds{E}[Y_x]\log q-\log \overline{y_0}-\KL(Y_x||\YYq)=R_p(x)
\]
be the associated KL-gap with $R_p(x)$ defined as in \eqref{eq:rp}. According to~\eqref{eq:deltadelta},
\begin{align}\label{eq:deltadeltatrunc}
	\overline\Delta_{\overline\delta}(x)=-\mathds{E}[Y_x]\log q-\log \overline{\alpha}-d\log\delta-\KL(Y_x||\YYq_\delta)=R_p(x)-d\log\overline\delta+d^x\log\overline\delta,
\end{align}
is the KL-gap associated to $\YYq_{\overline\delta}$, where $\overline{\alpha}$ is the normalizing factor of $\YYq_{\overline\delta}$.

In general, we have $\overline\Delta(1)=R_p(1)\gg 0$ and $\overline\Delta(x)\to 0$ exponentially fast when $x\to\infty$. Let ${\overline\delta}=\exp(-R_p(1)/d)$. Recalling~\eqref{eq:KLepsdelta}, this choice of $\delta$ leads to the upper bound
\begin{equation}\label{eq:KLepsdeltatrunc}
\KL(Y_x||\YYq_{\overline\delta})\leq -\log\overline{\alpha}-\mathds{E}[Y_x]\log q+R_p(1)-\overline{\eps}_{\overline\delta}(p),
\end{equation}
where $\overline{\eps}_{\overline\delta}(p)=\inf_{x\geq 1}\overline\Delta_{\overline\delta}(x)$ and $1/\overline{\alpha}={\overline\delta}+1/\overline{y_0}-1$.

Observe that $\overline\Delta_{\overline\delta}(1)=R_p(1)$ and $\overline\Delta_{\overline\delta}(x)\to R_p(1)$ still exponentially fast when $x\to\infty$. Experimentally, for $p\geq 0.6$ we have $\overline\Delta_\delta(x)\geq R_p(1)$ for all $x\geq 1$ (see Figure~\ref{fig:klgaptrunc}). Therefore, in such a case we have $\overline{\eps}_{\overline\delta}(p)=R_p(1)$ and so
\[
	\KL(Y_x||\YYq_{\overline\delta})\leq -\log\overline{\alpha}-\mathds{E}[Y_x]\log q
\]
with respective KL-gap
\begin{equation}\label{eq:deltaprimetrunc}
	\overline\Delta'_{\overline\delta}(x)= -\log\overline{\alpha}-\mathds{E}[Y_x]\log q-\KL(Y_x||\YYq_{\overline\delta})=\Delta_{\overline\delta}(x)-R_p(1)\geq 0
\end{equation}
In particular, we now have $\overline\Delta'_{\overline\delta}(1)=0$ and $\overline\Delta'_{\overline\delta}(x)\to 0$ when $x\to\infty$ exponentially fast, as desired.

Consequently, from~\eqref{eq:KLepsdeltatrunc} and Theorem~\ref{thm:duality} we obtain the following upper bound with the desired KL-gap for a large range of the replication parameter $p$.
\begin{thm}\label{thm:capdeltatrunc}
	We have
	\begin{equation}\label{eq:UBdeltatrunc}
		\Ca''_\mu(D_0)\leq \inf_{q\in(0,1)}(-\log\overline{\alpha}-\mu\log q)+R_p(1)-\overline\eps_{\overline\delta}(p),
	\end{equation}
	where ${\overline\delta}=\exp(-R_p(1)/d)$, $1/\overline{\alpha}={\overline\delta}+1/\overline{y_0}-1$, and $\overline\eps_{\overline\delta}(p)=\inf_{x\geq 1}\overline\Delta_{\overline\delta}(x)$.
\end{thm}

We now consider $\Yq$ from Section~\ref{sec:firstbound}. The reasoning is analogous to the previous case, so we skip most of it. In this case, we have
\begin{equation}\label{eq:deltax}
	\Delta(x)=-\log y_0-\mathds{E}[Y_x]\log q - \KL(Y_x||\Yq)=\mathds{E}\left[\log\frac{\binom{Y_x/p-1}{Y_x}}{\binom{Y_x+x-1}{Y_x}}\right].
\end{equation}
It can be observed that $\Delta(x)\to 1/2$ when $x\to \infty$.
In the cases where $\Delta(1)\geq 1/2$, we can follow the general reasoning previously described and set $\delta=\exp(-(\Delta(1)-1/2)/d)$. However, when $\Delta(1)<1/2$, we simply set $\delta=1$, i.e., we use the original distribution $\Yq$ (note that $\delta>1$ is not allowed). Therefore, in general we set $\delta=\min(\exp(-(\Delta(1)-1/2)/d),1)$.

We then have
\begin{equation}\label{eq:deltadeltaconv}
	\Delta_\delta(x)=-\log \alpha -\mathds{E}[Y_x]\log q-d\log\delta-\KL(Y_x||\Yq_\delta)=\Delta(x)-d\log\delta+d^x\log\delta\geq 0,
\end{equation}
where $1/\alpha=\delta+1/y_0-1$. If $\Delta(1)\geq 1/2$, this leads to the bound
\begin{equation}\label{eq:lineconvdelta}
	\KL(Y_x||\Yqd)\leq -\log\alpha-\mathds{E}[Y_x]\log q + \Delta(1)-1/2-\eps_{\delta}(p),
\end{equation}
where $\eps_{\delta}(p)=\inf_{x\geq 1}\Delta_\delta(x)$. Furthermore, in this case we have $\Delta_\delta(1)=\Delta(1)$ and $\Delta_\delta(x)\to \Delta(1)$ when $x\to\infty$, as before.

From experiments, for $0.35\leq p\leq 0.5$ we have $\Delta(1)>1/2$ and $\Delta_\delta(x)\geq \Delta(1)$ for all $x\geq 1$ (see Figure~\ref{fig:klgapconv}). This means that $\eps_{\delta}(p)=\Delta(1)$ in this case, and so
\begin{equation*}
	\KL(Y_x||\Yqd)\leq -\log\alpha-\mathds{E}[Y_x]\log q -1/2,
\end{equation*}
with associated KL-gap
\begin{equation*}
	\Delta'_\delta(x)=-\log\alpha-\mathds{E}[Y_x]\log q -1/2-\KL(Y_x||\Yq)=\Delta_\delta(x)-\Delta(1)\geq 0.
\end{equation*}
Observe that, similarly to previous cases, $\Delta'_\delta(1)=0$ and $\Delta'_\delta(x)\to 0$ when $x\to \infty$, as desired. Figure~\ref{fig:gapchange} showcases how the KL-gap changes for $p=1/2$ when we modify $\Yq$ at $y=0$ with our choice of $\delta$.

\begin{figure}
	\centering
	\includegraphics[width=0.7\textwidth]{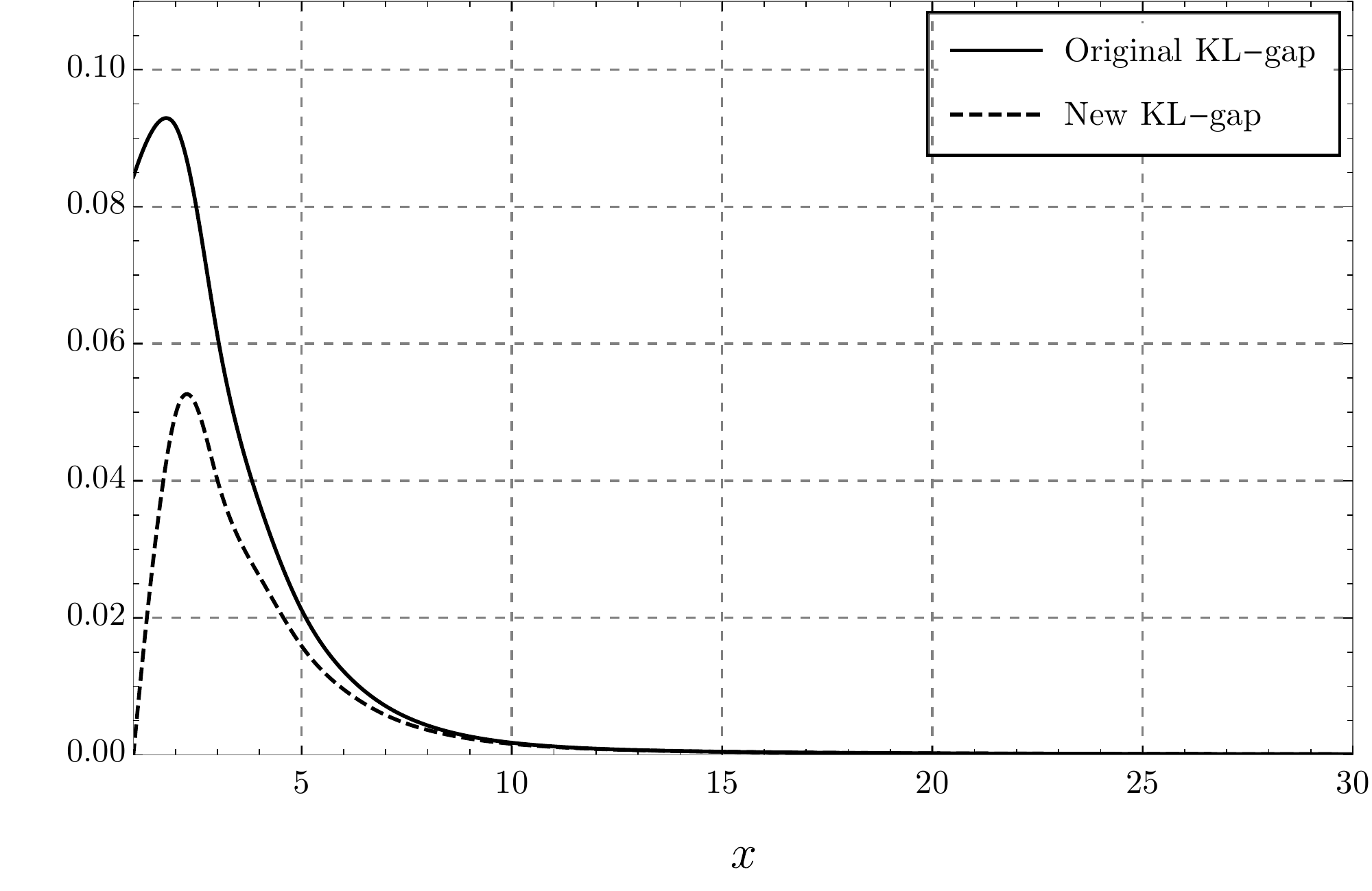}
	\caption{How the gap changes when the mass at $y=0$ is modified, as a function of $x$ for $p=1/2$ and $\Yq$ defined in Section \ref{sec:firstbound}. Black curve: The original KL-gap $\Delta(x)-1/2$. Dashed curve: The new  KL-gap $\Delta'_\delta(x)$ after fixing the mass at $y=0$ appropriately.}
	\label{fig:gapchange}
\end{figure}

From~\eqref{eq:lineconvdelta} and Theorem~\ref{thm:duality} we obtain the following upper bound with the desired KL-gap for a large range of the replication parameter $p$.
\begin{thm}\label{thm:capdeltaconv}
	We have
	\begin{equation}\label{eq:UBdeltaconv}
	\Ca''_\mu(D_0)\leq \inf_{q\in(0,1)}(-\log\alpha-\mu\log q)+\max(\Delta(1)-1/2,0)-\eps_{\delta}(p),
	\end{equation}
	where $\delta=\min(\exp(-(\Delta(1)-1/2)/d),1)$, $1/\alpha=\delta+1/y_0-1$, and $\eps_{\delta}(p)=\inf_{x\geq 1}\Delta_{\delta}(x)$.
\end{thm}

To conclude this section, we remark that the alternative choice $\delta=d$ for $\Yq$ leads to a better capacity upper bound than both Theorems~\ref{thm:capdeltatrunc}~and~\ref{thm:capdeltaconv} when $p$ is close to $1$. Interestingly, $\Yqd$ with $\delta=d$ corresponds \emph{exactly} to the inverse binomial distribution, which was designed independently for the deletion channel~\cite{Che17}. This choice of $\delta$ also leads to a simple, fully analytical proof that the capacity of the geometric deletion channel is bounded well away from $1$ when $p\to 1$ in Section~\ref{sec:boundsmalld}.

We argue that there is a natural justification behind the choice $\delta=d$. First, observe that we can extend the function $\Yq(\cdot)/y_0$ to $[0,\infty)$ in a natural way. Then, we have
\[
	\Yq(0)/y_0=\binom{-1}{0}=1.
\]
However, it is also the case that
\[
	\lim_{y\to 0^+}\Yq(y)/y_0=d<1.
\]
As a result, it follows that, in general, $\Yq(\cdot)/y_0$ is not right-continuous at $y=0$. We may choose $\delta$ so that $\Yqd(\cdot)/\alpha$ is right-continuous at $y=0$. It is immediate that the unique choice of $\delta$ that satisfies this is $\delta=d$.

\subsection{Capacity upper bounds for the geometric deletion channel}\label{sec:compbounds}

In this section, we analyze the capacity upper bounds we obtain for the geometric deletion channel by combining Corollary~\ref{cor:geombound} with the distributions designed in Sections \ref{sec:firstbound} and \ref{sec:trunc} and their modifications described in Section \ref{sec:fixmass0}.

It is easy to see that the capacity of the geometric deletion channel with duplication probability $p$ is upper bounded by the capacity of the deletion channel with deletion probability $d=1-p$. In fact, we can simulate the output of a geometric deletion channel via the output of the deletion channel by having the receiver replace every output bit by $D_1=1+D_0$ copies of it.

We will compare the bounds we obtain with the state-of-the-art capacity upper bounds for the deletion channel from \cite{RD15}. Furthermore, when $p=1/2$, the geometric deletion channel corresponds exactly to the binary replication channel studied by Mercier, Tarokh, and Labeau \cite{MTL12} with $p_d=p_t=1/2$. We will compare our bound with theirs for $p=1/2$.

For each $p\in [0,1)$, our bound is obtained by combining Corollary~\ref{cor:geombound} with Theorems~\ref{thm:capdeltatrunc}~and~\ref{thm:capdeltaconv}, and choosing, for each $\mu\geq 1$, the value of $q$ satisfying $\mathds{E}[\Yqd]=\mu$ (this is possible because both families of distributions grow like $\Theta(q^y/\sqrt{y})$).
\begin{coro}\label{cor:bound0geom}
	We have
	\begin{equation}\label{eq:boundconv}
	\Ca(D_0)\leq \sup_{q\in(0,1):\atop\mu_q\geq p/(1-p)}\frac{p(\eps_{\delta}(p)-d\log\delta-\log\alpha-\mu_q\log q)}{d(1+\mu_q)}
	\end{equation}
	with $\delta=\min(\exp(-(\Delta(1)-1/2)/d),1)$, and
	\begin{equation}\label{eq:boundtrunc}
	\Ca(D_0)\leq \sup_{q\in(0,1):\atop\overline\mu_q\geq p/(1-p)}\frac{p(\overline{\eps}_{\overline{\delta}}(p)-d\log\overline{\delta}-\log\overline{\alpha}-\overline{\mu}_q\log q)}{d(1+\overline{\mu}_q)}
	\end{equation}
	with $\overline{\delta}=\exp(-R_p(1)/d)$, where
	\begin{align*}
	1/\alpha&=\delta+\sum_{y=1}^\infty \binom{y/p-1}{y}q^y e^{-y h(p)/p},\\
	\mu_q&=\sum_{y=1}^\infty \alpha y\binom{y/p-1}{y}q^y e^{-y h(p)/p},\\
	1/\overline{\alpha}&=\overline{\delta}+\sum_{y=1}^\infty \frac{q^y e^{\Lambda_2(y)-\Lambda_1(y)-y\textnormal{Li}(1/(1+2p))-y h(p)/p}}{y!},\\
	\overline{\mu}_q&=\sum_{y=1}^\infty \frac{\overline{\alpha}yq^y e^{\Lambda_2(y)-\Lambda_1(y)-y\textnormal{Li}(1/(1+2p))-y h(p)/p}}{y!},
	\end{align*}
	with $\Lambda_1$ and $\Lambda_2$ as defined in \eqref{eq:lambda10geom} and \eqref{eq:lambda20geom}, respectively, $\eps_{\delta}(p)=\inf_{x\geq 1}\Delta_\delta(x)$ for $\Delta_\delta(x)$ defined in~\eqref{eq:deltadeltaconv}, and $\overline{\eps}_{\overline\delta}(p)=\inf_{x\geq 1}\overline\Delta_{\overline\delta}(x)$ for $\overline\Delta_{\overline\delta}(x)$ defined in~\eqref{eq:deltadeltatrunc}.
\end{coro}

Figure~\ref{fig:0geom} compares \eqref{eq:boundconv}, \eqref{eq:boundtrunc}, and the state-of-the-art capacity upper bound for the deletion channel from \cite{RD15}. Table \ref{table:comp0geom} contains, for selected values of $p$, a comparison between our best analytical upper bound at that point and the deletion channel capacity upper bound from \cite{RD15}. As mentioned at the end of Section~\ref{sec:fixmass0}, the choice $\delta=d$ works well for $p$ close to $1$. We include the bound induced by this choice of $\delta$ for large values of $p$ in Table~\ref{table:comp0geom}, appropriately identified. However, when $p$ is not very large, this bound worsens quickly, and so we opt not to include it in the plot.

Plots of the functions inside the suprema in \eqref{eq:boundconv} and \eqref{eq:boundtrunc} can be found in Figures~\ref{fig:curvesconv0geom} and \ref{fig:curvestrunc0geom}, respectively. Similarly to the geometric sticky and elementary duplication channels, these functions are concave.



Figures~\ref{fig:klgapconv}~and~\ref{fig:klgaptrunc} showcase the KL-gap attained by the distributions $\Yqd$ and $\YYq_{\overline{\delta}}$ from Sections~\ref{sec:firstbound}~and~\ref{sec:trunc}, respectively, with the choices of $\delta$ and $\overline{\delta}$ specified in Corollary~\ref{cor:bound0geom}. For the sake of comparison, Figures~\ref{fig:klgapconvdelta1}~and~\ref{fig:klgaptruncdelta1} show the original KL-gaps of the distributions $\Yq$ and $\YYq$. Observe that, in this case, both gaps at $x=1$ are noticeably larger than $0$. On the other hand, the gaps in Figures~\ref{fig:klgapconv}~and~\ref{fig:klgaptrunc} can be shifted down so that they are exactly (or at least close to) $0$ at $x=1$, and close to $0$ for large $x$. As can be seen, one can easily approximate $\eps_{\delta}(p)$ and $\overline{\eps}_{\overline{\delta}}(p)$ with high accuracy by numerically computing the KL-gap for a small number of values of $x$, especially for $\overline{\eps}_{\overline{\delta}}(p)$. This is due to the fact that 
$R_p(x)\to 0$ exponentially fast in $x$. 

If $p\in [0,0.5]$, the infimum in $\eps_{\delta}(p)$ is achieved at $x=1$ (see Figure~\ref{fig:klgapconv}), and the same holds for $\overline{\eps}_{\overline{\delta}}(p)$ if $p\in[0.6,1)$ (see Figure~\ref{fig:klgaptrunc}). Moreover, if $p\in [0.35,0.5]$, then $\Delta(1)\geq 1/2$. This means that the choices of $\delta$ and $\overline{\delta}$ in Corollary~\ref{cor:bound0geom} (which are derived in Section \ref{sec:fixmass0}) for $p\in [0.35,0.5]$ and $p\in[0.6,1)$, respectively, yield distributions $\Yqd$ and $\YYq_{\overline{\delta}}$ whose KL-gaps are exactly 0 at $x=1$ and converge to 0 quickly for large $x$.

In the case where $p=1/2$, the best known capacity upper bound was given in~\cite{MTL12}. They report a bound of $0.209092$ bits/channel use, obtained by employing a reduction from the original channel to a memoryless channel via the addition of commas between input runs which are never deleted (this same reduction was used in~\cite{DMP07}), coupled with clever numerical methods. Our analytical upper bound, which in particular employs a tighter reduction via Theorem~\ref{thm:red}, yields a bound of $0.168074$ bits/channel use.


\begin{table}[h]
	\caption{Comparison between the numerical capacity bounds for the deletion channel from~\cite{RD15} and the analytical upper bound from Corollary~\ref{cor:bound0geom} in bits/channel use. When $p$ is large, we include the better upper bound induced by the choice $\delta=d$.}\label{table:comp0geom}
	\centering
	\begin{tabular}{ |c|c|c| }
		\hline
		$p$ & Upper bound deletion~\cite{RD15} & Upper bound from Corollary~\ref{cor:bound0geom}\\ 
		\hline
		0.05 &  0.021 &0.021244\\ 
		\hline
		0.10 & 0.041 &0.041352\\ 
		\hline
		0.15 &  0.062 &0.061242\\ 
		\hline
		0.20 & 0.082&0.076981\\ 
		\hline
		0.25 &  0.103 &0.091134\\ 
		\hline
		0.30 &  0.123 & 0.104846\\ 
		\hline
		0.35 &  0.144 &0.119552\\ 
		\hline
		0.40 &  0.165 &0.135271\\ 
		\hline
		0.45 &  0.187 &0.151342\\ 
		\hline
		0.50 & 0.212 &0.168074\\ 
		\hline
		0.55 &  0.241 &0.186588\\ 
		\hline
		0.60 & 0.275 &0.204186\\ 
		\hline
		0.65 & 0.315 &0.234480\\ 
		\hline
		0.70 &  0.362 &0.262103\\ 
		\hline
		0.75 & 0.420 &0.269490\\ 
		\hline
		0.80 &  0.491 &0.271810\\ 
		\hline
		0.85 &  0.579 &0.270561\\ 
		\hline
		0.90 & 0.689 &0.275250 (0.310823 with $\delta=d$)\\ 
		\hline
		0.95 &  0.816 &0.337581 (0.326424 with $\delta=d$)\\ 
		\hline
		0.99 &  0.963 &0.769416 (0.338927 with $\delta=d$)\\ 
		\hline
	\end{tabular}
\end{table}

\begin{figure}
	\centering
	\includegraphics[width=0.7\textwidth]{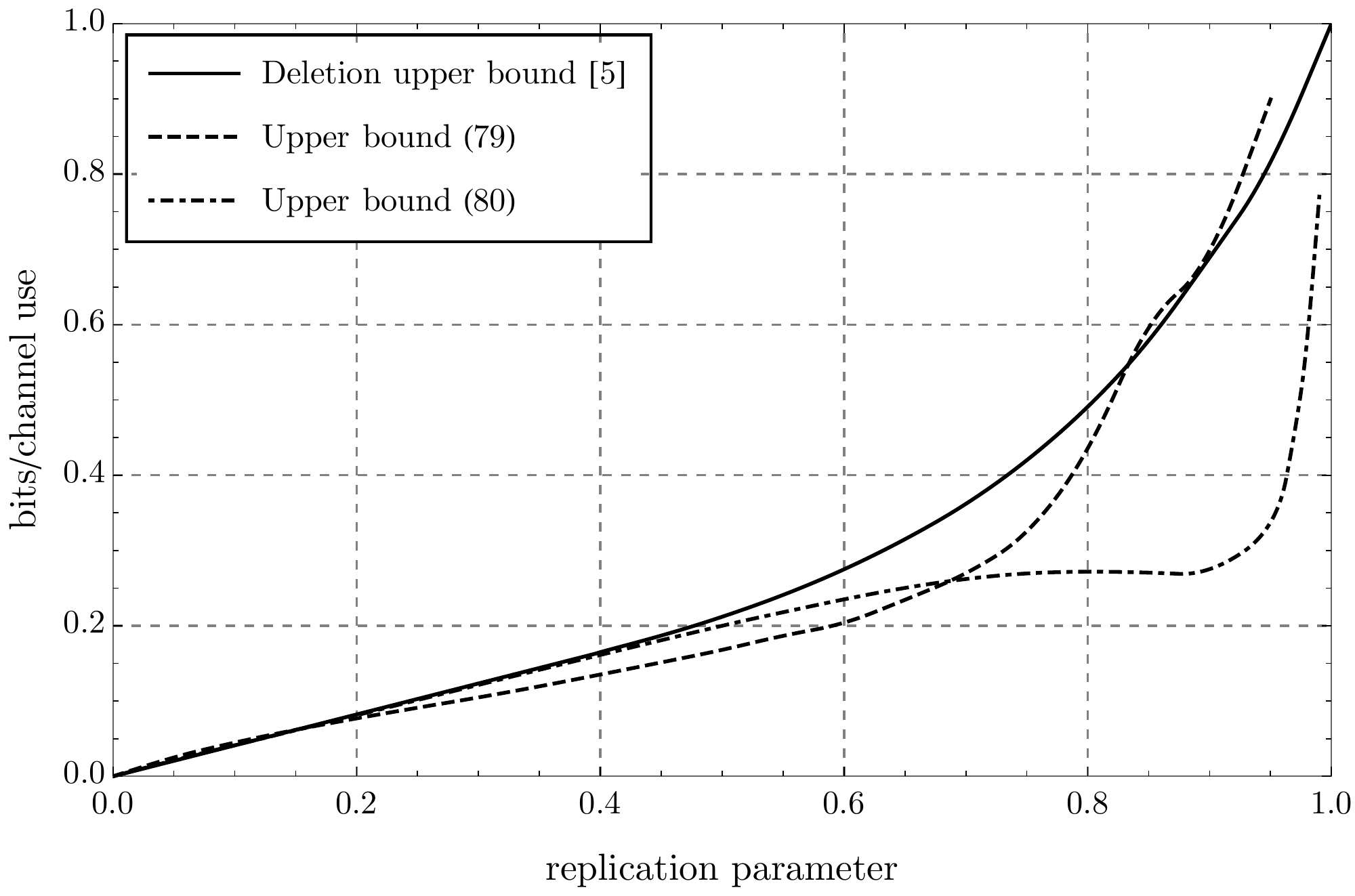}
	\caption{Plot of analytical upper bounds \eqref{eq:boundconv} and \eqref{eq:boundtrunc}, and the state-of-the-art deletion channel capacity upper bound from \cite{RD15}.}
	\label{fig:0geom}
\end{figure}

\begin{figure}
	\centering
	\includegraphics[width=0.7\textwidth]{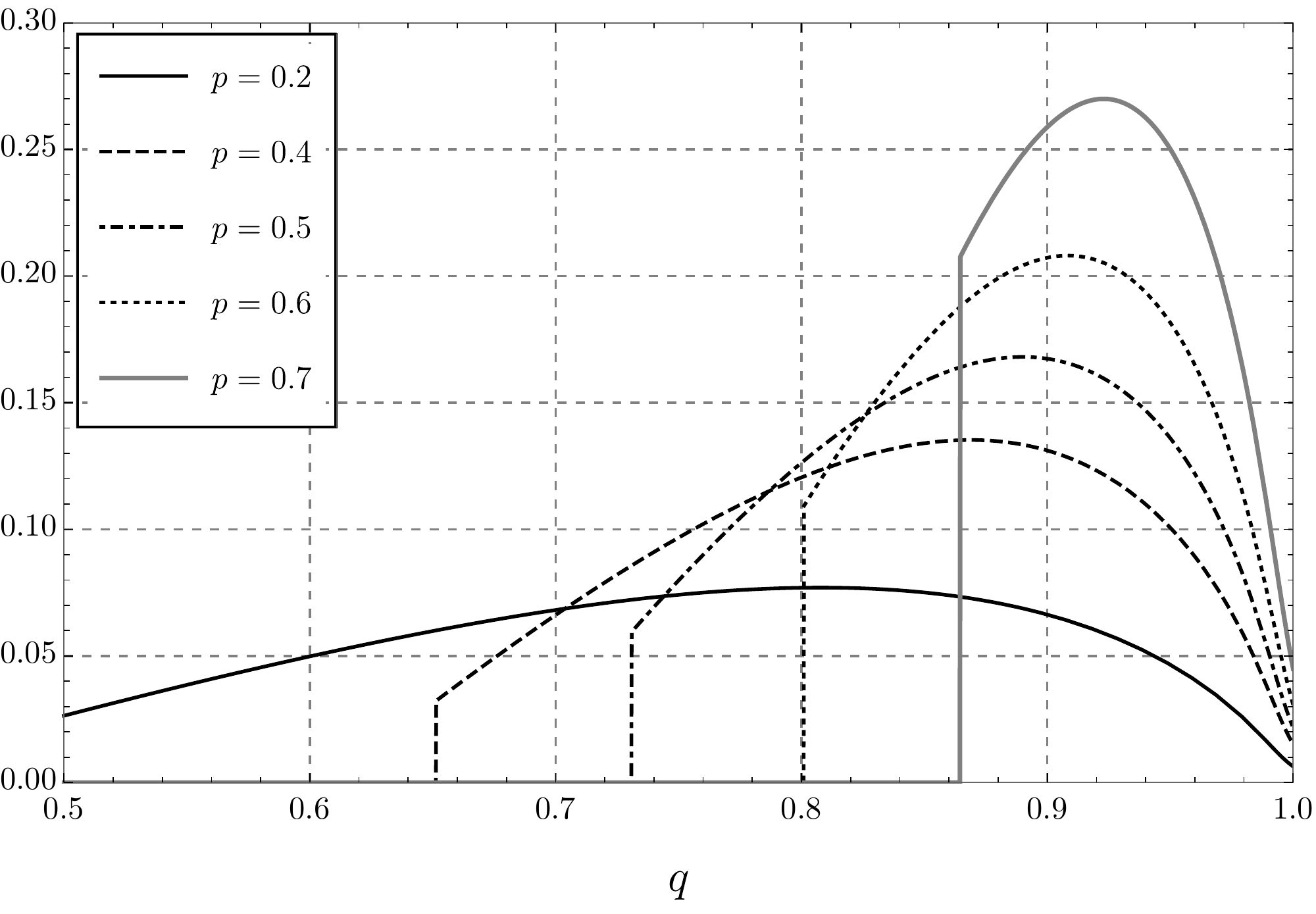}
	\caption{Function inside the supremum in \eqref{eq:boundconv} for some values of $p$. The zone where the function is zero corresponds to the cases where $\mathds{E}[\Yq_{\delta}]<\frac{p}{1-p}$.}
	\label{fig:curvesconv0geom}
\end{figure}

\begin{figure}
	\centering
	\includegraphics[width=0.7\textwidth]{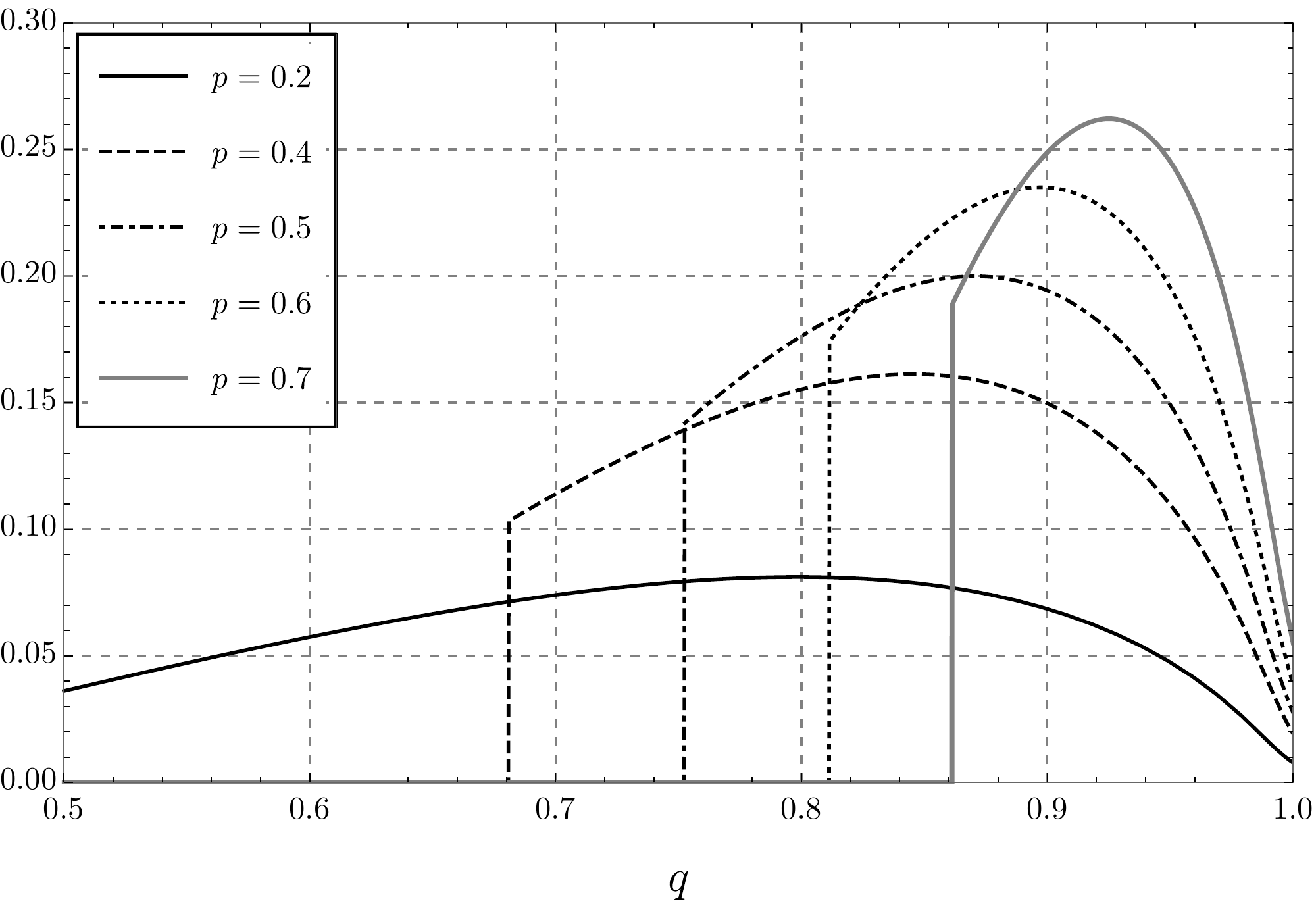}
	\caption{Function inside the supremum in \eqref{eq:boundtrunc} for some values of $p$. The zone where the function is zero corresponds to the cases where $\mathds{E}[\YYq_{\overline{\delta}}]<\frac{p}{1-p}$.}
	\label{fig:curvestrunc0geom}
\end{figure}

\begin{figure}
	\centering
	\includegraphics[width=0.7\textwidth]{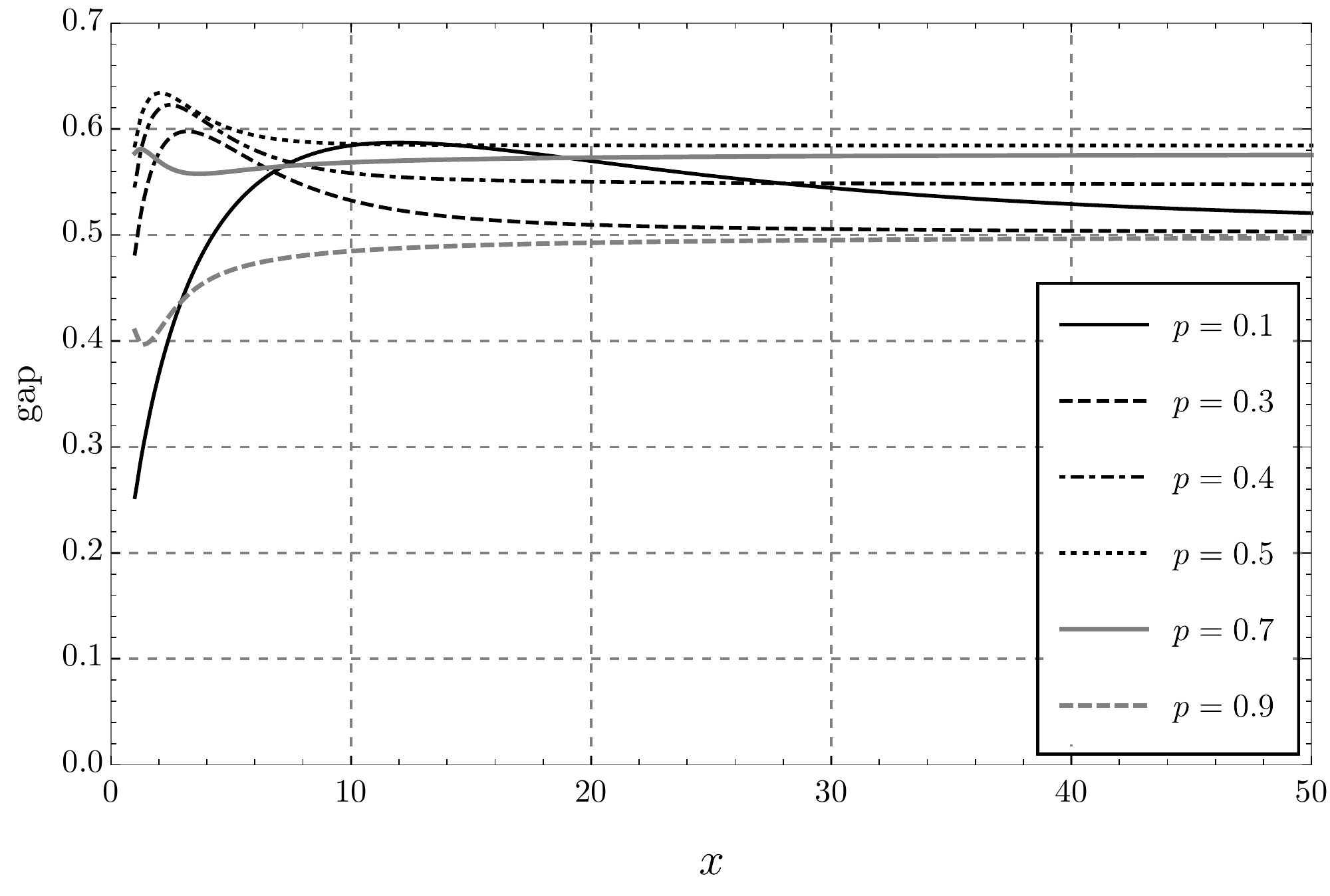}
	\caption{KL-gap $\Delta_\delta(x)$ (defined in~\eqref{eq:deltadeltaconv}) of the distribution $\Yqd$ from Section~\ref{sec:firstbound} with the choice of $\delta$ in Corollary~\ref{cor:bound0geom} plotted for $x\geq 1$ for some values of the replication parameter $p$.}
	\label{fig:klgapconv}
\end{figure}

\begin{figure}
	\centering
	\includegraphics[width=0.7\textwidth]{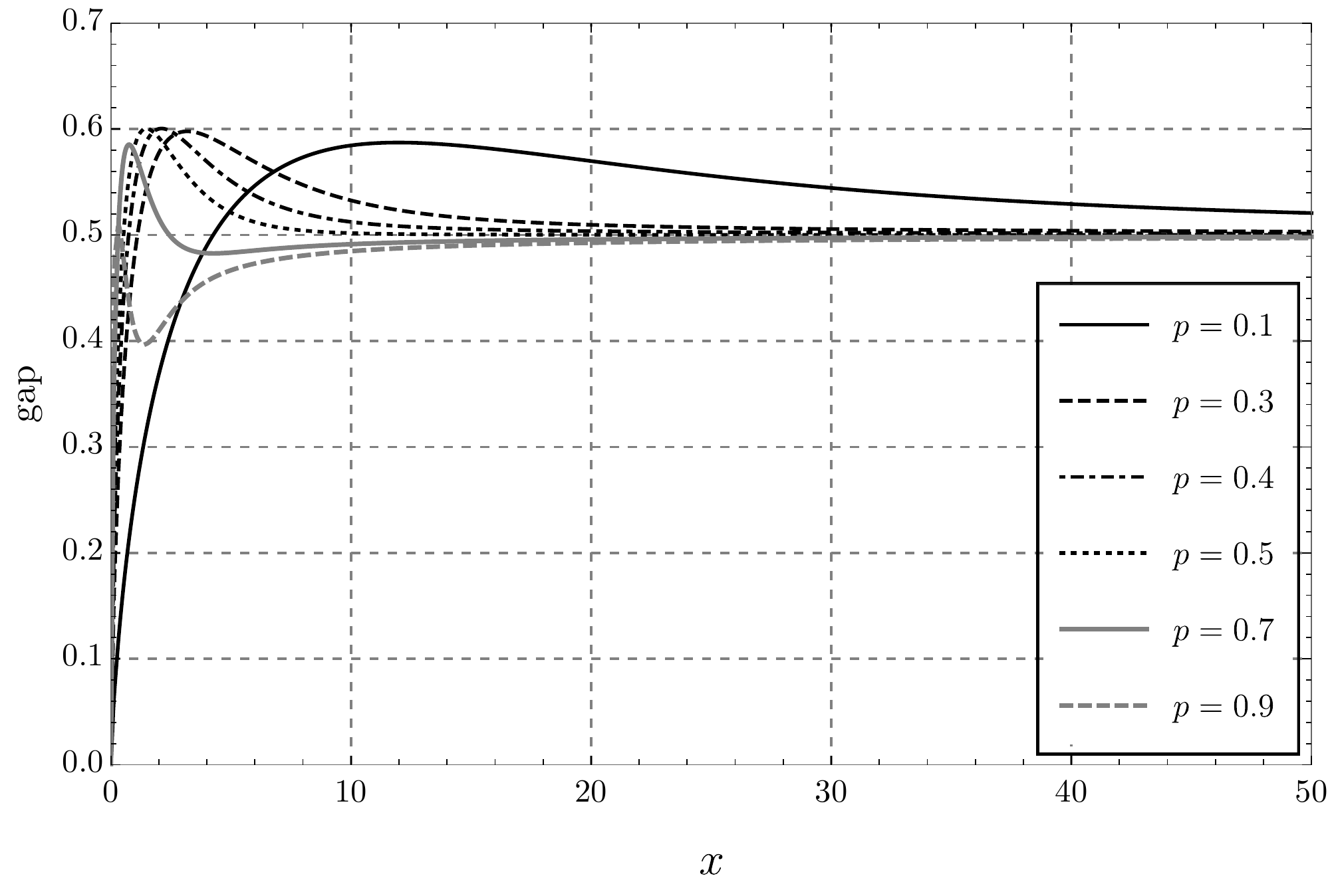}
	\caption{KL-gap $\Delta(x)$ (defined in~\eqref{eq:deltax}) of the distribution $\Yq$ from Section~\ref{sec:firstbound} plotted for $x\geq 0$ for some values of the replication parameter $p$.}
	\label{fig:klgapconvdelta1}
\end{figure}

\begin{figure}
	\centering
	\includegraphics[width=0.7\textwidth]{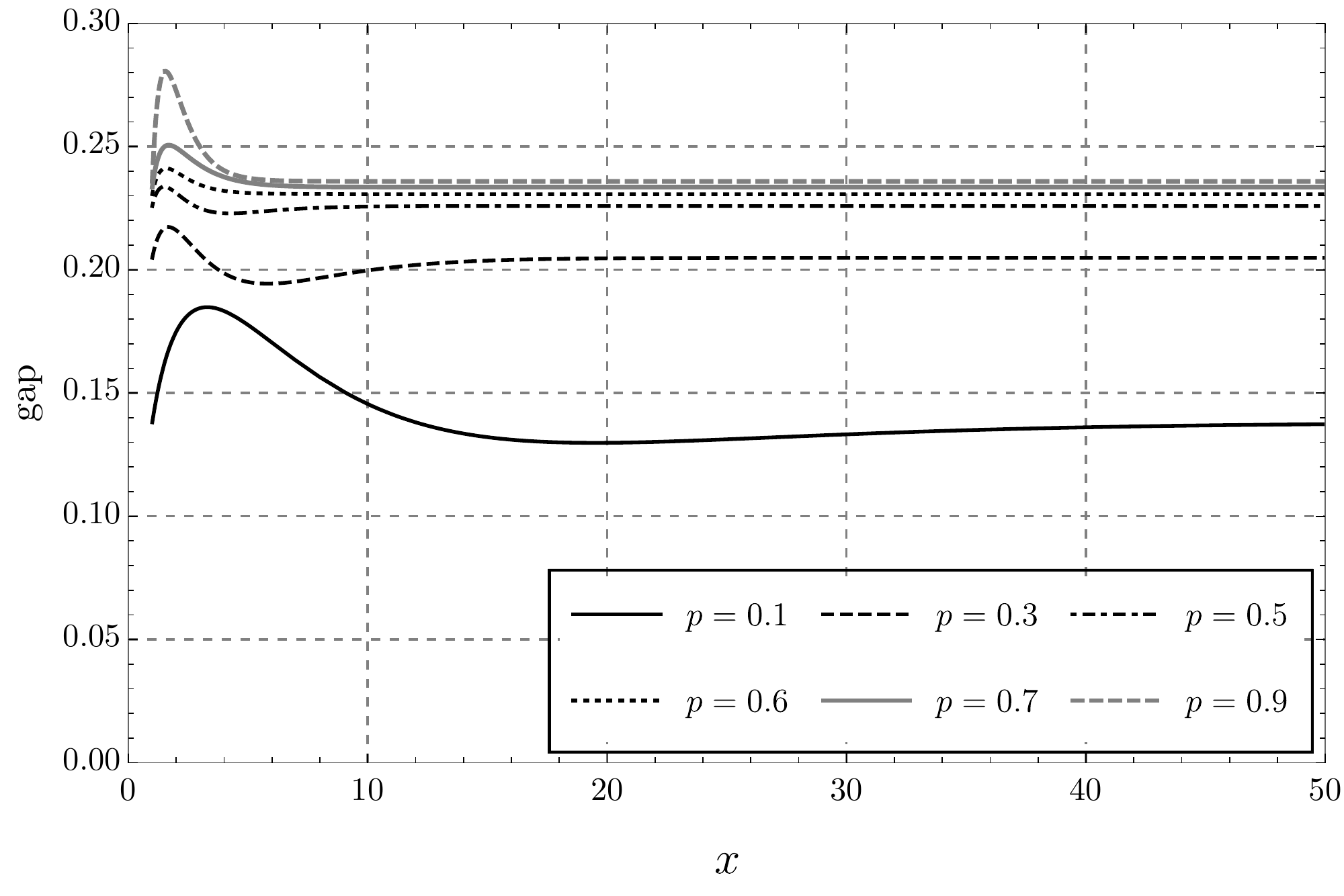}
	\caption{KL-gap $\overline{\Delta}_{\overline{\delta}}(x)$ (defined in~\eqref{eq:deltadeltatrunc}) of the distribution $\YYq_{\overline{\delta}}$ from Section~\ref{sec:trunc} with the choice of $\overline{\delta}$ in Corollary~\ref{cor:bound0geom} plotted for $x\geq 1$ for some values of the replication parameter $p$.}
	\label{fig:klgaptrunc}
\end{figure}

\begin{figure}
	\centering
	\includegraphics[width=0.7\textwidth]{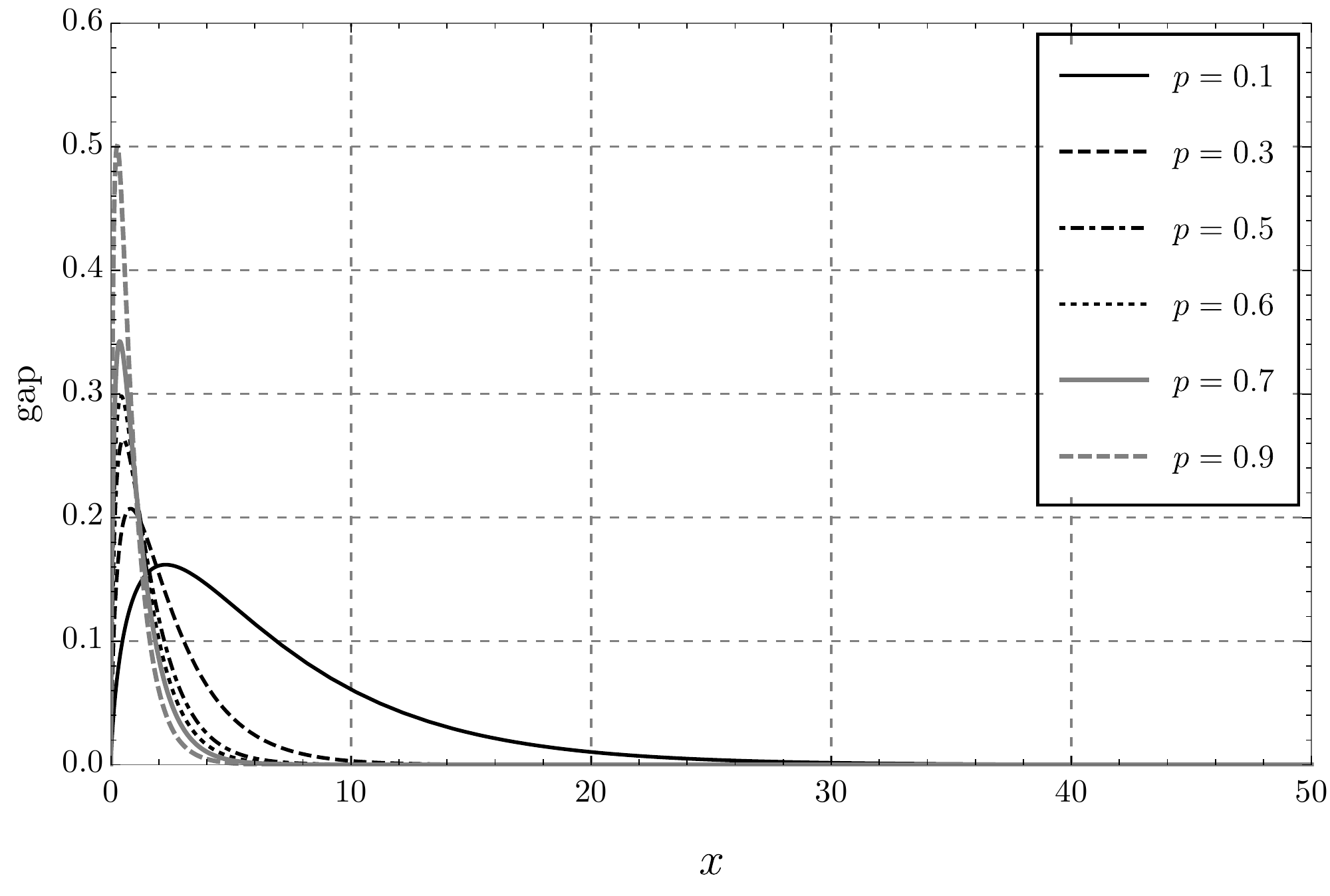}
	\caption{KL-gap $R_p(x)$ (defined in~\eqref{eq:rp}) of the distribution $\YYq$ from Section~\ref{sec:trunc} plotted for $x\geq 0$ for some values of the replication parameter $p$.}
	\label{fig:klgaptruncdelta1}
\end{figure}

\subsection{An elementary upper bound for large replication probability}\label{sec:boundsmalld}

Building up on results obtained in Sections \ref{sec:firstbound} and \ref{sec:fixmass0}, we give a simple and fully analytical proof that the capacity of the geometric deletion channel is at most $0.73$ bits/channel use for large replication parameter $p$.
\begin{thm}\label{thm:smalld}
	We have
	\[
	\Ca(D_0)\leq \frac{1}{2\log 2}+o(1) \quad\textrm{bits/channel use}
	\]
	when $p\to 1$, where $o(1)\to 0$ when $p\to 1$, and $\frac{1}{2\log 2}\approx 0.7214$.
\end{thm}
\begin{proof}
	For convenience, we define $d=1-p$. Combining Corollary~\ref{cor:geombound} and \eqref{eq:capUBdelta} instantiated with $\Yq$ defined in Section~\ref{sec:firstbound}, we conclude that
	\begin{equation}\label{eq:explicitbound}
	\Ca(D_0)\leq\frac{p}{d}\sup_{\mu\geq 1/d}\frac{\inf_{q,\delta}(-\eps_\delta(p)-d\log\delta-\log\alpha-(\mu-1)\log q)}{\mu},
	\end{equation}
	where the infimum is taken over all $q\in(0,1)$ and $\delta\in (0,1]$. Moreover, recalling \eqref{eq:epsdelta} and Lemma~\ref{lem:epsneg}, we have
	\[
		\eps_\delta(p)\geq \eps(p)\geq 0
	\]
for all $\delta\in(0,1]$ and $p\in(0,1)$. Therefore,
	\begin{equation}\label{eq:explicitbound2}
		\Ca(D_0)\leq\frac{p}{d}\sup_{\mu\geq 1/d}\frac{\inf_{q,\delta}(-d\log\delta-\log\alpha-(\mu-1)\log q)}{\mu},
	\end{equation}

 We set $\delta=d$, and begin by estimating $-\log \alpha$. Recall that $1/\alpha=\delta+1/y_0-1$. Then,
	\[
	1/\alpha=\delta+d(1/y_\mathsf{IB}-1)=d+d(1/y_\mathsf{IB}-1)=d/y_\mathsf{IB}.
	\]
	It is possible to bound $1/y_\mathsf{IB}$ according to \cite[Corollary~22]{Che17} for large $p$ as
	\[
	1/y_\mathsf{IB}\leq 1+\frac{1}{\sqrt{2d}}\left(\frac{1}{\sqrt{1-q}}-1\right),
	\]
	and so
	\[
	1/\alpha\leq d+\sqrt{d/2}\left(\frac{1}{\sqrt{1-q}}-1\right).
	\]
	Setting $q=1-d/2$ yields
	\begin{equation}\label{ineq:alpha}
	1/\alpha\leq d+\sqrt{d/2}\left(\frac{1}{\sqrt{d/2}}-1\right)=1+d-\sqrt{d/2}<1
	\end{equation}
	for $d<1/2$, which implies that $-\log \alpha< 0$. Taking into account \eqref{eq:explicitbound2} and setting $\delta=d$, $q=1-d/2$, we obtain the bound
	\begin{align*}
	\Ca(D_0)&\leq \frac{p}{d}\sup_{\mu\geq 1/d}\frac{-d\log d -\log\alpha-(\mu-1)\log q}{\mu}\\
	&\leq \sup_{\mu\geq 1/d}\frac{-d\log d -\log\alpha-(\mu-1)\log q}{\mu d}\\
	&\leq -d\log d-\frac{\log q}{d},
	\end{align*}
	where in the second inequality we used the fact that $p<1$, and in the third inequality we used the fact that $\mu d\geq 1$, $-\log \alpha<0$, and $-\frac{(\mu-1)\log q}{\mu d}\leq -\frac{\log q}{d}$.
	
	Recalling that $q=1-d/2$, we have $-\frac{\log q}{d}=\frac{1}{2}+o(1)$, where $o(1)\to 0$ when $d\to 0$ (equivalently, $p\to 1$). Finally, observe that $-d\log d=o(1)$ as well. This gives the desired bound in nats/channel use, and dividing it by $\log 2$ concludes the proof.
\end{proof}

\begin{remark}
	Note that choosing $\delta=d$ as we did in the proof is equivalent to choosing the inverse binomial distribution from \cite{Che17} as the candidate distribution $Y$.
\end{remark}



\section{Conclusions and future directions}

We derived analytical capacity upper bounds for sticky channels and a channel combining geometric replications and deletions, which we called the geometric deletion channel. 

Our bounds for sticky channels are extremely sharp if the duplication probability is not too large, and in fact improve upon the previously known numerical upper bounds for some values of the duplication probability. Moreover, our bounds are induced by distributions which achieve zero KL-gap in the framework of~\cite{Che17}. This is the first time such distributions have been designed for channels with synchronization errors. 

If the distributions with zero KL-gap were also valid channel output distributions, then we would have derived an exact expression for the capacity of the associated channels. However, this turns out not to be the case. A natural next step is to attempt to derive distributions which satisfy both these conditions. This most likely will require employing new techniques. It would also be interesting to find an example of a non-trivial repeat channel whose capacity can be determined exactly via our techniques. 

Another important path would be to determine the capacity of a finite version of the memoryless channels studied in Sections~\ref{sec:zerogapsticky}~and~\ref{sec:zerogapdupl}, where one only allows input $x\leq A$ for some fixed constant $A$.


We significantly improved upon the previous best capacity upper bounds for the geometric deletion channel. This was done by exploiting the fact that we can modify the mass of the underlying distribution at $y=0$ with ease in Section~\ref{sec:fixmass0}. Moreover, this observation also led to a simple, fully analytical proof of a non-trivial capacity upper bound for the geometric deletion channel with large duplication probability. In particular, we give a fully analytical proof that the capacity of this channel is bounded away from 1 when the replication parameter approaches 1. Such a bound was inaccessible via previous methods. A possible direction for future research is to obtain improved bounds for a continuous interval of the channel parameter with fully analytical proofs (both for the geometric deletion channel and other channels) by exploiting the technique from Section \ref{sec:fixmass0} in a more refined way.

\section{Acknowledgments}

We thank Iosif Pinelis for the simple proof of \eqref{eq:asympint} in the proof of Lemma~\ref{lem:asymp}. We also thank Hugues Mercier for providing us with the data points of the numerical bounds in~\cite{MTL12}.

\bibliographystyle{IEEEtran}
\bibliography{sticky-refs}

\begin{appendices}
	\section{The capacity of the Poisson-repeat channel for small deletion probability}\label{app:poisson}
	
	The Poisson-repeat channel is a repeat channel with replication distribution $D=\mathsf{Poi}_\lambda$, where $\mathsf{Poi}_\lambda$ denotes a Poisson distribution with expected value $\lambda$, i.e.,  
	\[
	D(y)=\frac{e^{-\lambda}\lambda^y}{y!},\quad y=0,1,2,\dots
	\]
	In this appendix, we show that the capacity of the Poisson-repeat channel with parameter $\lambda$ converges to $1$ when $\lambda\to\infty$. This regime corresponds to the setting where the expected number of bit replications grows to infinity, or, equivalently, the deletion probability $D(0)=e^{-\lambda}$ converges to $0$.
	
	Before we prove the desired result, we need the following concentration bound for the Poisson distribution. This bound is a corollary of Bennett's inequality~\cite{Can17}.\footnote{Alternatively, one can obtain a concentration bound for $\mathsf{Poi}_\lambda$ by considering a Chernoff bound for $\mathsf{Bin}_{n,\lambda/n}$ and noting that $\mathsf{Bin}_{n,\lambda/n}$ converges to $\mathsf{Poi}_\lambda$ in distribution when $n\to\infty$.}
	\begin{lem}\label{lem:concpoisson}
		We have
		\[
		\Pr[(1-\delta)\lambda\leq \mathsf{Poi}_\lambda\leq(1+\delta)\lambda]\geq 1- 2\exp\left(-\frac{\delta^2\lambda}{4}\right)
		\]
		for $0\leq\delta\leq 1$.
	\end{lem}
	
	The following lemma states that we can approximate the true channel input from its output in edit distance with high probability. We denote the edit (Levenshtein) distance between two strings $x$ and $y$ by $\mathsf{ED}(x,y)$.
	\begin{lem}\label{lem:poissondecode}
		Given $0<\eps<1$ and $\lambda$ large enough, the following holds. Let $Y$ be the output of the Poisson-repeat channel with parameter $\lambda$ given some fixed arbitrary $n$-bit string $x$ as input. Then, we can obtain $\hat{x}$ from $Y$ such that $\mathsf{ED}(x,\hat{x})\leq \eps n$ with probability $1-o_n(1)$ as $n\to\infty$.
	\end{lem}
	\begin{proof}
		We begin by describing how we obtain $\hat{x}$ from $Y$. Given some string $s$, we call a maximal consecutive sequence of bits with the same value in $s$ a \emph{run}. Let $Y_i$ denote the $i$-th run in $Y$. Then, the $i$-th run of $\hat{x}$ is obtained by writing down the bit value that appears in $Y_i$ exactly $[|Y_i|/\lambda]$ times, where $[w]$ denotes the closest integer to $w$.
		
		We now upper bound $\mathsf{ED}(x,\hat{x})$. Let $L_i$ denote the number of times $x_i$ is replicated in $Y$ for $i=1,2,\dots, n$. Each bit $x_i$ contributes at most $2|L_i/\lambda-1|$ to $\mathsf{ED}(x,\hat{x})$. Therefore, we have
		\[
		\mathsf{ED}(x,\hat{x})\leq 2\sum_{i=1}^n \left|L_i/\lambda-1\right|.
		\]
		It now remains to show that $2\sum_{i=1}^n |L_i/\lambda-1|\leq \eps n$ with probability $1-o_n(1)$ as $n\to\infty$ if $\lambda$ is large enough.
		
		With some hindsight, let $\delta=\eps/8$. Note that, by Lemma~\ref{lem:concpoisson}, we have
		\begin{equation}\label{eq:outlength}
		(1-\delta)\lambda n\leq |Y|\leq (1+\delta)\lambda n
		\end{equation}
		with probability at least $1-2\exp\left(-\frac{\delta^2 \lambda n}{4}\right)=1-o_n(1)$. This is because $|Y|$ is distributed according to $\mathsf{Poi}_{\lambda n}$.
		
		Recall that $L_i$ denotes the number of times $x_i$ is replicated in $Y$. We say $i$ is \emph{$\delta$-good} if
		\[
		(1-\delta)\lambda \leq L_i\leq (1+\delta)\lambda,
		\]
		and we say that $i$ is \emph{$\delta$-bad} otherwise. By Lemma~\ref{lem:concpoisson}, the probability that $i$ is $\delta$-good is at least $1-2\exp\left(-\frac{\delta^2 \lambda}{4}\right)$. A standard application of the Chernoff bound implies that, with probability $1-o_n(1)$, at most an $\eps_{\textrm{bad}}=4\exp\left(-\frac{\delta^2 \lambda}{4}\right)$ fraction of $i$'s are $\delta$-bad.
		
		From the definition of $\delta$-good and $\eps_{\textrm{bad}}$ it follows that with probability $1-o_n(1)$ we have
		\begin{equation}\label{eq:li}
		(1-\eps_{\textrm{bad}})(1-\delta)\lambda n\leq\sum_{i:\textrm{ $i$ is $\delta$-good}}L_i\leq (1+\delta)\lambda n.
		\end{equation}
		
		Combining \eqref{eq:li}, the fact that $|Y|=\sum_{i=1}^n L_i$, and \eqref{eq:outlength}, with probability $1-o_n(1)$ it holds that
		\[
		\sum_{i:\textrm{ $i$ is $\delta$-bad}}L_i\leq (1+\delta)\lambda n - (1-\eps_{\textrm{bad}})(1-\delta)\lambda n\leq\alpha\lambda n
		\]
		for $\alpha=2\delta+\eps_{\textrm{bad}}$. As a result,
		\begin{align}
		\sum_{i:\textrm{ $i$ is $\delta$-bad}}|L_i-\lambda| &\leq \sum_{i:\textrm{ $i$ is $\delta$-bad}}(L_i+\lambda)\nonumber\\
		&\leq (\alpha+\eps_{\textrm{bad}})\lambda n\label{eq:deltabad}
		\end{align}
		holds with probability $1-o_n(1)$.
		
		From the previous observations, with probability $1-o_n(1)$ we have
		\begin{align*}
		2\sum_{i=1}^n |L_i/\lambda-1|&=\frac{2}{\lambda}\sum_{i=1}^n |L_i-\lambda|\\
		&=\frac{2}{\lambda}\left(\sum_{\substack{i:\\\textrm{ $i$ is $\delta$-good}}}|L_i-\lambda|+\sum_{\substack{i:\\\textrm{ $i$ is $\delta$-bad}}}|L_i-\lambda|\right)\\
		&\leq \frac{2}{\lambda}(\delta\lambda n+(\alpha+\eps_{\textrm{bad}})\lambda n)\\
		&=2(\delta+\alpha+\eps_{\textrm{bad}})n\\
		&\leq \eps n
		\end{align*}
		if $\lambda$ is large enough, as desired. The first inequality follows from the definition of $\delta$-good and \eqref{eq:deltabad}. The second inequality holds because $\eps_{\textrm{bad}}\leq \eps/16$ if $\lambda$ is large enough (recall that $\eps$ is a constant), and thus, in this case,
		\[
		2(\delta+\alpha+\eps_{\textrm{bad}})\leq 2(\eps/8+2\eps/8+\eps/16+\eps/16)=\eps.
		\]
		
	\end{proof}
	
	Let $\Ca(\lambda)$ denote the capacity of the Poisson-repeat channel with parameter $\lambda$. We are now ready to prove the following result.
	\begin{thm}
		We have
		\[
		\lim_{\lambda\to\infty}\Ca(\lambda)=1.
		\]
	\end{thm}
	\begin{proof}
		We prove this result by showing that for any $\delta>0$ we have $\Ca(\lambda)\geq 1-\delta$ provided that $\lambda$ is large enough.
		
		It is easy to show that there exist families of codes which correct an $\eps$ fraction of deletions and insertions with rate approaching $1$ as $\eps\to 0$. In fact, almost optimal explicit constructions of efficiently decodable codes of this type are known~\cite{Hae18,CJLW18}.
		
		Fix $\delta>0$ and let $\eps>0$ be small enough so that there exists a code $C$ of rate $1-\delta$ which corrects an $\eps$ fraction of deletions and insertions. Furthermore, let $\lambda$ be large enough so that Lemma~\ref{lem:poissondecode} holds with this specific choice of $\eps$.
		
		Consider the following coding scheme: To transmit a message through the Poisson-repeat channel with parameter $\lambda$, the sender transmits a codeword $c\in C$. By Lemma~\ref{lem:poissondecode}, the receiver can recover $\hat{c}$ such that $\mathsf{ED}(c,\hat{c})\leq \eps n$ with probability $1-o_n(1)$. Since $C$ corrects an $\eps$ fraction of deletions and insertions, it follows that the receiver can recover $c$ from $\hat{c}$ via unique decoding. This implies that $\Ca(\lambda)\geq 1-\delta$ whenever $\lambda$ is large enough. Since $\delta$ was arbitrary, we have the desired result.
	\end{proof}

\end{appendices}

\end{document}